\PassOptionsToPackage{unicode}{hyperref}
\PassOptionsToPackage{hyphens}{url}
\documentclass[
  fleqn]{article}
\usepackage{lmodern}
\usepackage{amssymb,amsmath}
\usepackage{ifxetex,ifluatex}
\ifnum 0\ifxetex 1\fi\ifluatex 1\fi=0 
  \usepackage[T1]{fontenc}
  \usepackage[utf8]{inputenc}
  \usepackage{textcomp} 
\else 
  \usepackage{unicode-math}
  \defaultfontfeatures{Scale=MatchLowercase}
  \defaultfontfeatures[\rmfamily]{Ligatures=TeX,Scale=1}
\fi
\IfFileExists{upquote.sty}{\usepackage{upquote}}{}
\IfFileExists{microtype.sty}{
  \usepackage[]{microtype}
  \UseMicrotypeSet[protrusion]{basicmath} 
}{}
\makeatletter
\@ifundefined{KOMAClassName}{
  \IfFileExists{parskip.sty}{%
    \usepackage{parskip}
  }{
    \setlength{\parindent}{0pt}
    \setlength{\parskip}{6pt plus 2pt minus 1pt}}
}{
  \KOMAoptions{parskip=half}}
\makeatother
\usepackage{xcolor}
\IfFileExists{xurl.sty}{\usepackage{xurl}}{} 
\IfFileExists{bookmark.sty}{\usepackage{bookmark}}{\usepackage{hyperref}}
\hypersetup{
  pdftitle={A non-convex regularization approach for stable estimation of loss development factors},
  pdfauthor={Himchan Jeong; Hyunwoong Chang; Emiliano A. Valdez},
  hidelinks,
  pdfcreator={LaTeX via pandoc}}
\usepackage[margin=1in]{geometry}
\usepackage{graphicx,grffile}
\makeatletter
\def\maxwidth{\ifdim\Gin@nat@width>\linewidth\linewidth\else\Gin@nat@width\fi}
\def\maxheight{\ifdim\Gin@nat@height>\textheight\textheight\else\Gin@nat@height\fi}
\makeatother
\setkeys{Gin}{width=\maxwidth,height=\maxheight,keepaspectratio}
\makeatletter
\def\fps@figure{htbp}
\makeatother
\usepackage{soul,color,colortbl}
\usepackage{bbm}
\usepackage{multirow}
\usepackage{amssymb}
\usepackage{amsthm}
\usepackage{amsmath}
\usepackage{xcolor}
\usepackage{setspace}
\usepackage{hhline}
\usepackage{url}
\usepackage{caption}

\newtheorem{theorem}{Theorem}
\newtheorem{lemma}{Lemma}
\usepackage{pdflscape}
\newcommand{\blandscape}{\begin{landscape}}
\newcommand{\elandscape}{\end{landscape}}
\usepackage{algorithm,algorithmic}

\usepackage{booktabs}
\usepackage{longtable}
\usepackage{array}
\usepackage{multirow}
\usepackage{wrapfig}
\usepackage{float}
\usepackage{colortbl}
\usepackage{pdflscape}
\usepackage{tabu}
\usepackage{threeparttable}
\usepackage{threeparttablex}
\usepackage[normalem]{ulem}
\usepackage{makecell}
\usepackage{xcolor}
\usepackage[]{natbib}
\title{A non-convex regularization approach for stable estimation of loss
development factors}
\author{Himchan Jeong\footnote{himchan\_\href{mailto:jeong@sfu.ca}{\nolinkurl{jeong@sfu.ca}};
  Corresponding author, Department of Statistics \& Actuarial Science,
  Simon Fraser University, 8888 University Drive, Burnaby, BC, V5A 1S6,
  Canada.} \and Hyunwoong Chang\footnote{\href{mailto:hwchang1201@tamu.edu}{\nolinkurl{hwchang1201@tamu.edu}};
  Department of Statistics, Texas A\&M University, College Station, TX,
  77843-3143, USA.} \and Emiliano A. Valdez\footnote{\href{mailto:emiliano.valdez@uconn.edu}{\nolinkurl{emiliano.valdez@uconn.edu}};
  Department of Mathematics, University of Connecticut, 341 Mansfield
  Road, Storrs, CT, 06269-1009, USA.}}
\date{}

\begin{document}
\maketitle

\hypertarget{abstract}{%
\subsection*{Abstract}\label{abstract}}
\addcontentsline{toc}{subsection}{Abstract}

In this article, we apply non-convex regularization methods in order
{to obtain} stable estimation of loss development factors
{in insurance claims reserving}. Among the non-convex regularization
methods, we focus on the use of {the} log-adjusted absolute deviation
(LAAD) penalty and provide discussion on optimization of LAAD penalized
{regression} model, which {we prove} to converge with a coordinate
descent algorithm under mild conditions. This has the advantage of
obtaining a consistent estimator for the regression coefficients
{while allowing for} the variable selection, which is linked to the
stable estimation of loss development factors. We calibrate our proposed
model using a multi-line insurance dataset from a property and casualty
insurer where we observed reported aggregate loss along accident years
and development periods.
{When compared to other regression models, our LAAD penalized regression model provides very promising results.}

\vspace{1cm}

\textbf{Keywords:} Insurance reserving, log-adjusted absolute deviation
(LAAD) penalty, loss development, non-convex penalization, robust
estimation, variable selection.

\vspace{6cm}

\pagebreak

\section{Introduction} \label{sec:intro}

{The} chain ladder method, as an industry benchmark with theoretical
{foundation as discussed in} \citet{mack1993chainladder} and
\citet{mack1999chainladderse}, has been widely used to determine the
development pattern of reported or paid claims. {Despite its}
prevalence, we need to {address} some {relevant} issues in
{the} estimation of development factors for mature years with
{this} chain ladder method. In general, we expect that cumulative
reported loss amount {increases through time} whereas the magnitude
of {the} development decreases. However, it is possible that loss
development patterns with some run-off triangles may not follow
{these expected} patterns
{since we have only a few data points in the top-right} corner and
estimation of parameters {that} depends on those data points becomes
unstable, due to triangular or trapezoidal shape of aggregated claim
data {as mentioned in} \citet{renshaw1989chain}. Therefore, we need
to consider {some ways} to estimate the development factors for
mature years with more stability.

To deal with the aforementioned issue of stable estimation of
development factors for mature years, one can apply
{the regularization} method or penalized regression in loss
development models. Today, there is a rich literature of using
penalization
\footnote{{In this article, we have interchangeably used the terms penalization and regularization; both are standard terms in the literature.}}
in {the} regression framework. The first penalization method
introduced is ridge regression, developed by \citet{hoerl1970ridge}. By
adding an \(L_2\) penalty term on the least squares, they showed that it
is possible to have smaller mean squared error when there is severe
multicollinearity in a given dataset. Nevertheless, ridge regression has
merely the shrinkage property {but} not the property of variable
selection. To tackle this problem, \citet{tibshirani1996lasso} suggested
LASSO {(least absolute shrinkage and selection operator)} using an
\(L_1\) penalty term on the least squares, and showed that this enables
us to {perform} variable selection. This method leads to dimension
reduction as well. Despite the simplicity of the proposed method, there
has been a great deal of work done to extend the LASSO framework. For
example, \citet{park2008bayesian} {solidified} Tibshirani's work by
providing a Bayesian interpretation on LASSO. Although LASSO has the
variable selection property, the estimates derived from LASSO regression
are inherently biased
{and may over-shrink the retained variables \mbox{\citep{hastie2009elements, hastie2015sparsity}}}.

There have been some meaningful approaches so that we obtain both
variable selection and consistency of estimates. For example,
\citet{fan2001scad} proposed smoothly clipped absolute deviation (SCAD)
penalty, derived by assuming continuously differentiable penalty
function to achieve three properties such as (i) consistency of the
estimate as the true value of the parameter increases, (ii) variable
selection, and (iii) the continuity of the calculated estimates.
Although SCAD penalty has the above-mentioned properties, it naturally
{leads to} a non-convex optimization so that it loses desirable
properties of convex optimization problems. \citet{zhang2010mcp}
proposed minimax concave penalty (MCP), which minimizes the maximum
concavity subject to {an} unbiasedness feature. Moreover,
\citet{lee2010hal} and \citet{armagan2013gdp} proposed {the}
log-adjusted absolute deviation (LAAD) penalty, which is derived by
{imposing a} hyperprior {distribution} on the tuning parameter
\(\lambda\), as proposed in \citet{park2008bayesian}.
{The term ``LAAD penalty'' has been coined in this paper.}

The {concept} of penalized regression in the actuarial literature is
not quite new.
{Indeed, the well-known Whittaker-Henderson mortality graduation method introduced a penalty to balance fit and smoothness of the observed data. See \mbox{\citet{chan1982wh}}.}
Yet there has only been a few relevant other works in the field.
\citet{williams2015elastic} applied {the} elastic net penalty, which
is a combination of \(L_1\) and \(L_2\) penalties, on a dataset with
over 350 initial covariates to enhance insurance claims prediction.
\citet{nawar2016} used LASSO for detecting possible interaction between
covariates, which are used in claims modeling. \citet{yin2016efficient}
{used a version of non-convex penalty for efficient estimation of Erlang mixture model with a group medical insurance claims data.}
Recently, \cite{mcguire2018lasso} proposed the use of \(L_1\)
penalization for stable estimation of loss development factors.

In this paper, we explore the use of LAAD {regression} that enables
us to {implement} variable selection
{and to obtain smoothness and stability} while maintaining
consistency of the estimators. {These properties inspired} us to
apply LAAD penalty in {a} cross-classified regression model
{to estimate} loss development in claims reserving.
{We also provide theoretical discussions on the convergence of LAAD penalized models based on the coordinate descent algorithm, and describe sufficient conditions for convergence, which are well-satisified in the applications for claims reserving. In calibrating various competing models, we use data drawn from a property and casualty (P\&C) insurer with two lines of business; the claims reported are conveniently expressed as loss triangles. We compare the estimated loss development factors of our proposed model with a straightforward cross-classified model and cross-classified models with LASSO and non-convex penalties, including LAAD. All these models are comparably explained in the section on estimation and prediction; we also discuss the validation measures for comparison purposes. Among these various models, our LAAD regression model performed remarkably well, which produced estimates of loss development factors with reasonable patterns and consistent with a priori knowledge. It also provides more improved prediction of reserve estimates. We additionally examined its usefulness in insurance ratemaking and this is summarized in the appendix.}

This paper has been organized as follows. In Section
\ref{sec:blassolaad}, we introduce the construction of the LAAD penalty
and provide sufficient condition {for} the convergence of
{the coordinate descent algorithm that forms the foundation to}
calibrate a LAAD regression model. In Section \ref{sec:simul}, we show
the efficacy of the proposed method {using a} simulation study. In
Section \ref{sec:insurance}, we explore possible use of LAAD penalty in
insurance
{claims reserving using an insurer's dataset with multi-line reported loss triangles. We discuss the results of the estimation and evaluate predictions for the various models examined}.
We {conclude} in Section \ref{sec:conclude}.

\section{LAAD penalization model and optimization} \label{sec:blassolaad}

\subsection{Derivation and properties of LAAD penalty} \label{sub:bayeslasso}

According to \citet{park2008bayesian}, we may interpret LASSO in a
Bayesian framework as follows: \[
\begin{aligned}
&Y|\beta \sim N(X\beta,\sigma^2 I_n), \quad \beta_i | \lambda \sim \text{Laplace}(0,1/\lambda),
\end{aligned}
\] where \(\beta\) is a vector of size \(p\) with each component having
a density function
\(p(\beta_i|\lambda)=\frac{\lambda}{2}e^{-\lambda{|\beta_i|}}\), for
\(\beta_i \in \mathbb{R}\). According to their specification, we may
express the likelihood and the log-likelihood for \(\beta\),
respectively, as

\begin{equation}\label{eq:1}
\begin{aligned}[b]
&L(\beta|y,X,\lambda) \  \propto \ \exp \left(-\frac{1}{2\sigma^2}\left[\sum_{i=1}^n  (y_i-X_i\beta)^2 \right]-\lambda||\beta||_1 \right) \ \text{and} \\
&\ell(\beta|y,X,\lambda) = -\frac{1}{2\sigma^2}\left[\sum_{i=1}^n  (y_i-X_i\beta)^2\right]-\lambda||\beta||_1+\text{Constant} .
\end{aligned}
\end{equation}

\citet{park2008bayesian} suggested two {approaches} to choose the
optimal \(\lambda\) in Equation (\ref{eq:1}). One is the use of point
estimate by cross-validation, and the other is the use of a `hyperprior'
distribution for \(\lambda\).

Now, consider the following distributional assumptions \[
\begin{aligned}
&Y|\beta \sim N(X\beta,\sigma^2 I_n), \quad \beta_j | \lambda_j \sim \text{Laplace}(0,1 / \lambda_j), \quad \text{and} \quad \lambda_j|r \overset{i.i.d.}{\sim} \text{Gamma}(r/\sigma^2-1,1).
\end{aligned}
\] In other words, the hyperprior of \(\lambda_j\) follows a gamma
distribution with density
\(p(\lambda_j|r)=\lambda_j^{\frac{r}{\sigma^2}-2}e^{-\lambda_j}/\Gamma(\frac{r}{\sigma^2}-1)\).
This implies that we have: \begin{equation}\label{eq:2}
\begin{aligned}[b]
&L(\beta,\lambda_1,\ldots,\lambda_p|y,X,r) \  \propto \ \exp \left(-\frac{1}{2\sigma^2}[\sum_{i=1}^n  (y_i-X_i\beta)^2] \right)\times \prod_{j=1}^{p}\exp \left( -\lambda_j\left[|\beta_j|+1\right] \right)\lambda_j^{\frac{r}{\sigma^2}-1}, \\
&L(\beta|y,X,r)=\int L(\beta,\lambda|y,X,r)d\lambda \  \propto \ \exp \left(-\frac{1}{2\sigma^2}[\sum_{i=1}^n  (y_i-X_i\beta)^2] \right)\times \prod_{j=1}^{p} \left( 1+|\beta_j| \right)^{-\frac{r}{\sigma^2}}, \\
&\ell(\beta|y,X,r)= -\frac{1}{2\sigma^2} \left(\sum_{i=1}^n  (y_i-X_i\beta)^2 +2r \sum_{j=1}^{p}\log(1+|\beta_j|)\right) + \text{Constant}.
\end{aligned}
\end{equation}

As a result, the log-likehood in Equation (\ref{eq:2}) allows us to have
the following formulation of our penalized least squares problem. This
gives rise to what we call the log-adjusted absolute deviation (LAAD)
penalty function: \[
||\beta||_L=\sum_{j=1}^p \log (1+|\beta_j|),
\] so that \[
\begin{aligned}
\hat{\beta} &=\ \mathop\mathrm{argmin}_{\beta} \frac{1}{2}||y-X\beta ||^2+r\sum_{j=1}^{p} \log(1+|\beta_j|).
\end{aligned}
\] Such derivation of LAAD penalty has been used in the statistical
literature including an application on non-convex paths
\citep{mazumder2011sparsenet} and health care \citep{wang2019novel}.

To {further understand} the characteristics of a model with LAAD
penalty, consider a simple example when \(p=1\) and \(||X||=1\). In this
case, optimization of \(\ell\) in Equation (\ref{eq:2}) is reduced to a
univariate case so that it is enough to solve the following:
\begin{equation}\label{eq:3}
\hat{\theta}_j=\mathop\mathrm{argmin}_{\theta_j} \frac{1}{2}(z_j-\theta_j)^2+r\log(1+|\theta_j|),
\end{equation} where \(z = X'y\).

\begin{theorem} \label{thm:1}
Let us set $\ell(\theta|r,z)=\frac{1}{2}(z-\theta)^2+r\log(1+|\theta|)$. Then the corresponding minimizer will be given as
$\hat{\theta}=\theta^*  \cdot \mathbbm{1}_{ \{|z|\geq z^*(r) \vee r \}}$, where
$$
\theta^* = \frac{1}{2} \left[z+ \text{sgn}(z)\left(\sqrt{(|z|-1)^2+4|z|-4r}-1\right) \right],
$$
and $z^*(r)$ is the unique solution of
$$
\Delta(z|r) = \frac{1}{2}(\theta^{*})^2-\theta^{*}z+r\log(1+|\theta^{*}|)=0.
$$
\end{theorem}

\begin{proof}
See Appendix A.
\end{proof}

Note that
\(\theta^{*}= \frac{z}{2} + \text{sgn}(z)\left[\frac{\sqrt{(|z|+1)^2-4r}-1}{2}\right] \simeq \frac{z}{2} + \text{sgn}(z)\left[\frac{(|z|+1)-1}{2}\right]=z\)
when \(|z|\) is large enough, which means \(\theta^{*}\) converges to
\(z\) when \(|z| \rightarrow \infty\). Therefore, by using LAAD penalty,
we obtain an optimizer which has {properties of} variable selection,
consistency, and continuity as demonstrated in \citet{armagan2013gdp}.

Figure \ref{fig:1} provides graphs that describe the behavior of the
obtained optimizer derived with different penalization. The first graph
is the behavior of the optimizer derived with \(L_2\) penalty, which is
also called ridge regression. In this case, as previously alluded, it
has no variable selection property but it only shrinks the magnitude of
the estimates. The second graph is the behavior of the optimizer derived
with \(L_1\) penalty, which is the basic LASSO. In this case, we see
that although it has {the} variable selection property (if value of
\(\beta\) is small enough, then \(\hat{\beta}\) becomes 0), the
discrepancy between the true \(\beta\) and \(\hat{\beta}\) remains
constant even when the true \(|\beta|\) is very big. Finally, the third
graph shows the behavior of the optimizer derived with the proposed LAAD
penalty. One can see that not only {does} the given optimizer have
the variable selection property, but also \(\hat{\beta}\) converges to
\(\beta\) as \(|\beta|\) increases.

\begin{figure}
\centering
\includegraphics{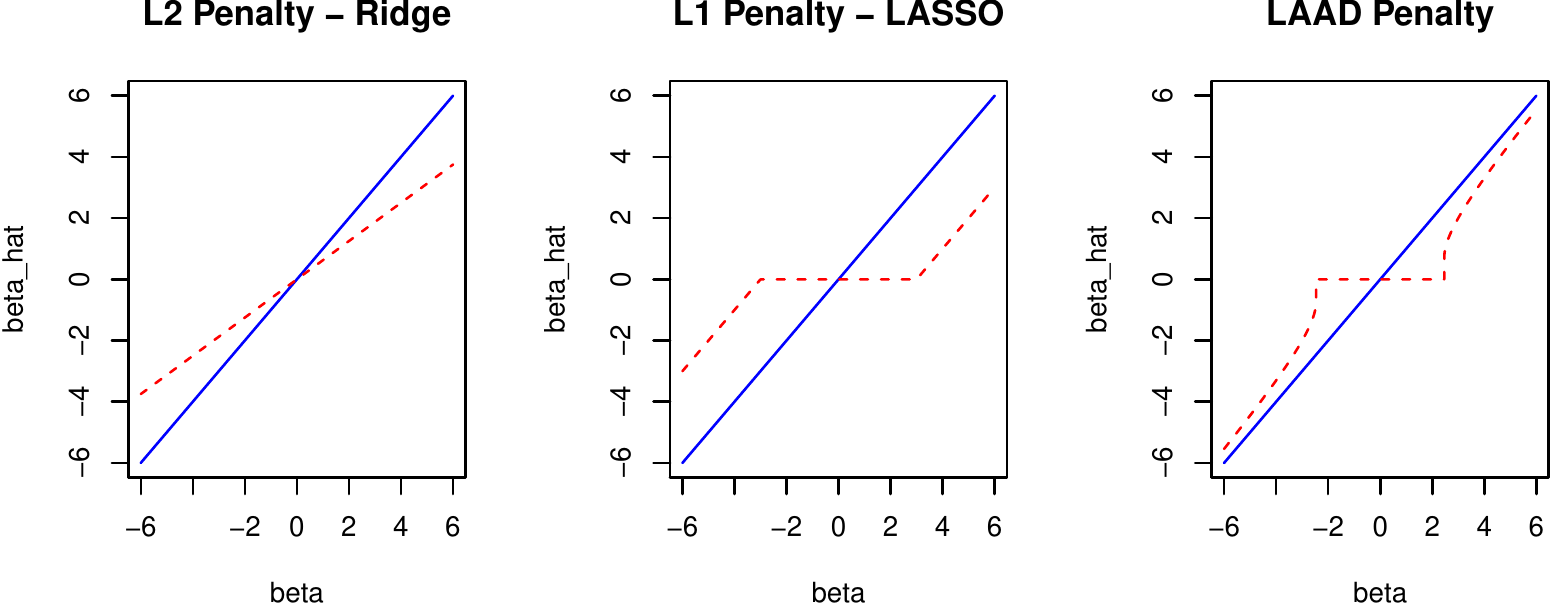}
\caption{\label{fig:1}Estimate behavior for several penalty functions}
\end{figure}

Figure \ref{fig:2} illustrates the constraint regions implied by each
penalty. It is well known that the constraint regions defined by \(L_2\)
penalization is a \(p\)-dimensional circle, whereas the constraint
regions defined by \(L_1\) penalization is a \(p\)-dimensional diamond.
We observe that in both cases of \(L_2\) and \(L_1\) penalization, the
constraint regions are convex, which implies {that} we entertain the
properties of convex optimization. In the case of the constraint implied
by LAAD penalty, the region is non-convex, which is inevitable to obtain
both consistency of the estimates and variable selection property
{as in the case of Bridge or $L_q$ penalization with $q \in (0,1)$. However, it is known that there is no closed form solution in the univariate case of $L_q$ penalized linear regression when $q \in (0,1)$ \mbox{\citep{knight2000asymptotics}} so that it is difficult to apply $L_q$ penalty $(0<q<1)$ on high-dimensional problems with the coordinate descent algorithm, let alone analyze convergence.}

\begin{figure}
\centering
\includegraphics{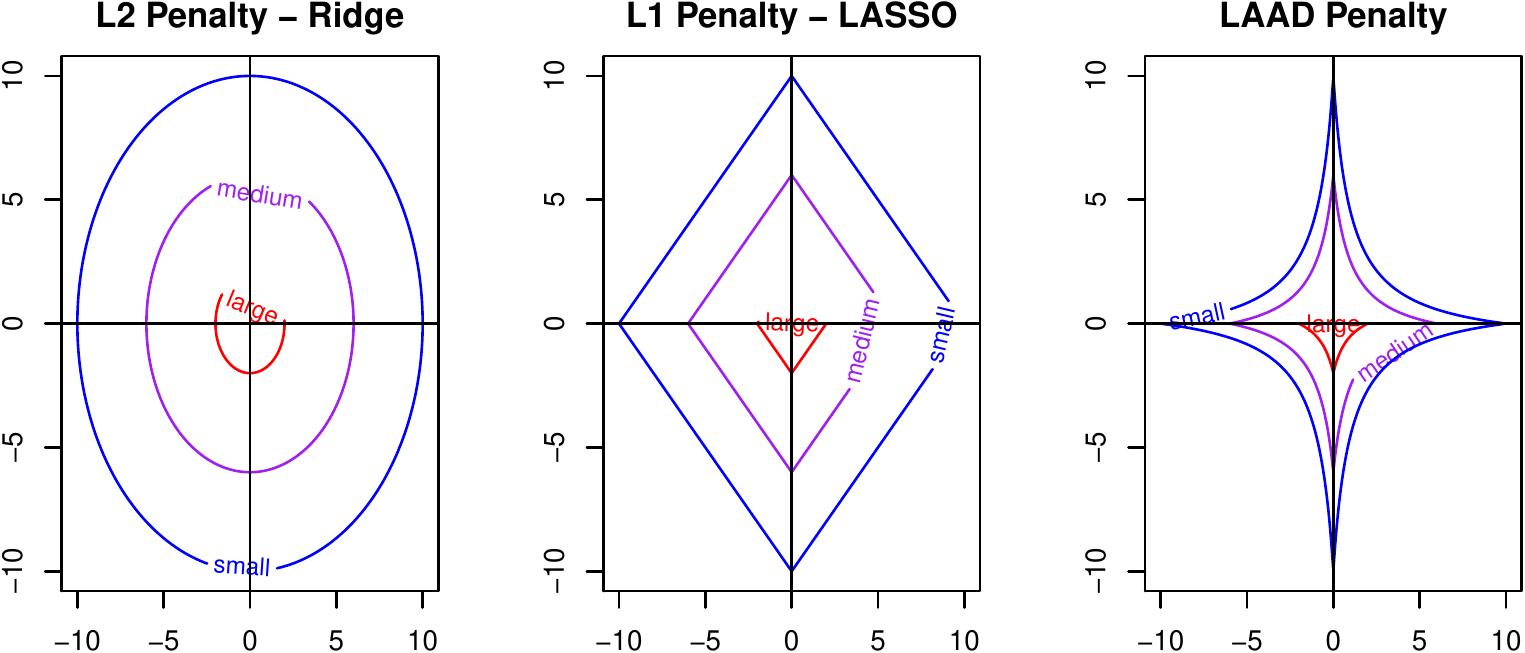}
\caption{\label{fig:2}Constraint regions for several penalties}
\end{figure}

It is also possible to compare the behavior of LAAD penalty with SCAD
penalty and MCP. According to \citet{fan2001scad} and
\citet{zhang2010mcp}, one can write down the penalty functions (and
their derivatives) in the univariate case, assuming \(\theta \geq 0\),
as follows: \begin{equation} \label{eq:penalty}
\begin{aligned}
& p_{LASSO}(\theta;\lambda) = \lambda \theta, \quad p'_{LASSO}(\theta;\lambda) = \lambda,  \\
& p_{MCP}(\theta;\lambda,\gamma) = \int_0^\theta \lambda (1-x/\gamma \lambda )_+dx, \quad p'_{MCP}(\theta;\lambda,\gamma) = \lambda (1-\theta/\gamma \lambda )_+, \\
& p'_{SCAD}(\theta;\lambda,a) = \lambda \{ I(\theta \leq \lambda) + \frac{(a\lambda-\theta)_+}{(a-1)\lambda}I(\theta > \lambda)  \}, \ \text{and} \\
& p_{LAAD}(\theta;\lambda) = \lambda \log (1+\theta), \quad p'_{LAAD}(\theta;\lambda) = \lambda \left( \frac{1}{1+\theta} \right).
\end{aligned}
\end{equation}

From above, it is straightforward to see that \[
\lim_{\theta \rightarrow \infty} p'_{SCAD}(\theta;\lambda,a)=\lim_{\theta \rightarrow \infty}p'_{MCP}(\theta;\lambda,\gamma)=\lim_{\theta \rightarrow \infty}p'_{LAAD}(\theta;\lambda)=0.
\] This implies that the marginal effect of penalty converges to 0 as
the value of \(\theta\) increases and hence, the magnitude of distortion
on the estimate becomes negligible as the true coefficient gets larger
when we use either SCAD, MCP, or LAAD penalty. However, we see that
\(\lim_{\theta \rightarrow \infty} p'_{LASSO}(\theta;\lambda) = \lambda\),
which means that the magnitude of distortion on the estimate is the same
even in the case when the true coefficient is very large.

On the other hand, if we let \(\theta\) {go} to 0, one can see that
\[
\lim_{\theta \rightarrow 0+} p'_{SCAD}(\theta;\lambda,a)=\lim_{\theta \rightarrow 0+}p'_{MCP}(\theta;\lambda,\gamma)=\lim_{\theta \rightarrow 0+}p'_{LAAD}(\theta;\lambda)=\lim_{\theta \rightarrow 0+}p'_{LASSO}(\theta;\lambda)=\lambda,
\] which implies that for SCAD, MCP, and LAAD penalty, the magnitude of
penalization is the same as LASSO when the true value of \(\theta\) is
very small. Therefore, we verify that SCAD, MCP, and LAAD penalties have
the same property of variable selection as LASSO, when the true
\(\theta\) is small enough.

{To summarize, it can be deduced that the use of LAAD penalty rewards sparsity more so than LASSO penalty. Although the overall degree of penalization is determined by cross-validation in both cases, penalization is evenly applied on all coefficients in the case of LASSO, while penalization is more severe for smaller coefficients in the case of LAAD. The landscape of penalization can be quite different from each other. Therefore, use of LAAD penalized regression results in less shrinkage on the estimated coefficients but more variable selection.}

\subsection{Implementation in general case and convergence analysis} \label{sub:optim}

Estimating parameters from given penalized least squares is an
optimization problem. Since an analytic solution is obtained in the case
of univariate penalized least squares, one can implement an algorithm
for optimization. For example, in obtaining \(\hat\beta\) in the
multivariate case, we may apply {a} coordinate descent algorithm
proposed by \citet{luo1992coordinate}, which starts with an initial
estimate and then successively optimizes along each coordinate or blocks
of coordinates. The algorithm is explained in detail as follows:

\begin{algorithm}[!ht]
\label{alg1}
\caption{A coordinate descent algorithm for LAAD penalty}
\begin{algorithmic}
\STATE \textbf{Input:} Training dataset $X=(X_1, \ldots, X_p),\, y$, initial parameter $\beta^{(0)} = \hat{\beta}^{OLS}$, a function $\hat{\theta}$ defined on Theorem 1, regularization parameter $r$, tolerance $\epsilon > 0$.
\end{algorithmic}
\begin{algorithmic}[1]
\STATE \textbf{while} $||\beta^{(t)} - \beta^{(t-1)}|| > \epsilon$ :
\STATE \quad Residual $\leftarrow y - \sum_{j=1}^p X_j\beta^{(t-1)}$
\STATE \quad \textbf{For} $j$ in $1:p$ :
\STATE \quad \quad Residual $\leftarrow$ Residual + $X_j\beta_j^{(t-1)}$
\STATE \quad \quad $z_{(t-1, j)} \leftarrow  X_j^\top$Residual 
\STATE \quad \quad $\beta_{(j)}^{(t)} \leftarrow \hat{\theta}(z_{(t-1, j)}, r)$ 
\STATE \quad \quad Residual $\leftarrow$ Residual - $X_j\beta_j^{(t)}$
\STATE \textbf{Output:} $\hat{\beta}^{LAAD} = \beta^{(t)}$
\end{algorithmic}
\end{algorithm}

Interestingly, {although the} use of LAAD penalty has been explored
in the statistics literature, a thorough analysis on the convergence of
{a} coordinate descent algorithm for LAAD regression model is still
scarce. \cite{mazumder2011sparsenet} found that application of {a}
coordinate descent algorithm to LAAD penalized model ``can produce
multiple limit points (without converging) - creating statistical
instability in the optimization procedure'', though they did not provide
a sufficient condition which assures statistical stability in the
optimization procedure. In this regard, here we provide a sufficient
condition so that {the} coordinate descent algorithm converges with
our optimization problem. To {prove} convergence, we need to
introduce the concepts of quasi-convex and hemivariate. A function is
hemivariate if a function is not constant on any interval which belongs
to its domain. A function \(f\) is quasi-convex if \[
f(x+\lambda d) \leq \max \{ f(x), \, f(x+d) \}, \ \text{for all } x, d \text{ and } \lambda\in [0,1].
\] An example of a function which is quasi-convex and hemivariate is
\(f(x)=\log(1+|x|)\).

The following lemma is useful for obtaining a sufficient condition that
our optimization problem converges with {the} coordinate descent
algorithm.

\begin{lemma} \label{lem:1}
Suppose a function $\ell: \mathbb{R}^p \rightarrow \mathbb{R}$ is defined as follows:
$$
\begin{aligned}
\ell(\beta_1, \ldots, \beta_p) = \frac{1}{2}||y-X\beta ||^2+r\sum_{j=1}^{p} \log(1+|\beta_j|) \\
\end{aligned}
$$
and $||X_j||=1$ for all $j=1,\ldots, p$. If $r \leq 1$, then $\ell_j: \beta_j \mapsto \ell(\beta_1, \ldots, \beta_p)$ is both quasi-convex and hemivariate for all $j=1,\ldots, p$.
\end{lemma}

\begin{proof}
See Appendix B.
\end{proof}

Note that the sum of quasi-convex functions may not be quasi-convex.
Therefore, although both \(||y-X\beta ||^2\) and \(\log(1+|\beta_j|)\)
are quasi-convex functions as functions of \(\beta_j\) on
\(\mathbb{R}\), it does not assure that
\(\ell(\beta_1, \ldots, \beta_p)\) is a quasi-convex function for each
coordinate. Intuitively, a continuous function on \(\mathbb{R}\) is
quasi-convex if and only if it has unique local minimum. In this regard,
\(f(x) = \log(1+|x|) + 0.01(10-x)^2\) is a quasi-convex function since
\[
f'(x) = \begin{cases}
\frac{1}{1+x} + 0.02(x-10) > 0; & x > 0, \\
\frac{-1}{1-x} + 0.02(x-10) < 0; & x < 0, \\
\end{cases}
\] and \(f(x)\) has unique local (indeed, global) minimum at \(x=0\).

However, if \(g(x) = 7\log(1+|x|) + (3-x)^2\), then \(g(x)\) is not
quasi-convex since such function has two local minima, \(0\) and
\(1+1/\sqrt{2}\). Figure \ref{fig:3} {provides} visualization of
these two examples.

\begin{figure}
\centering
\includegraphics{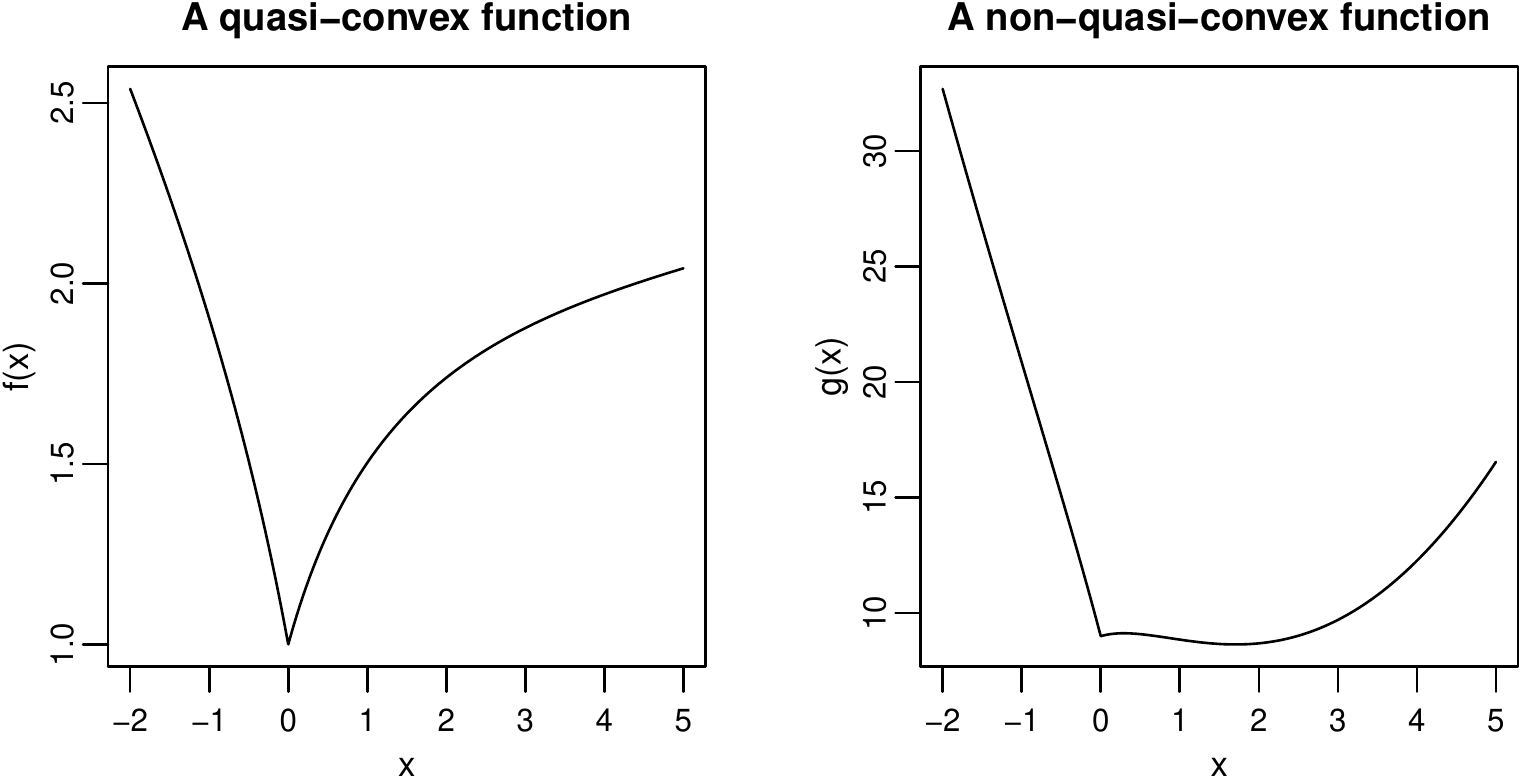}
\caption{\label{fig:3}An example of quasi-convex function and
non-quasi-convex function}
\end{figure}

Therefore, \(r \leq 1\) is the critical condition in the proof of Lemma
\ref{lem:1}, and it leads to the following theorem which assures the
convergence of coordinate descent algorithm for LAAD regression models.

\begin{theorem} \label{thm:2}
If $||X_j||=1$ for all $j=1,\ldots, p$ and $r \leq 1$, then the solution from {the} coordinate descent algorithm with function $l: \mathbb{R}^p \rightarrow \mathbb{R}$ converges to $\hat{\beta}$ where
$$
\begin{aligned}
\hat{\beta} &=\ \mathop\mathrm{argmin}_{\beta} \frac{1}{2}||y-X\beta ||^2+r\sum_{j=1}^{p} \log(1+|\beta_j|).
\end{aligned}
$$
\end{theorem}

\begin{proof}
According to Theorem 5.1 of \cite{tseng2001cdescent}, it suffices to show that 1) $||y-X\beta ||^2$ is continuous on $\mathbb{R}^p$, 2) $\log(1+|\beta_j|)$ is lower semicontinuous, and 3) $\ell_j: \beta_j \mapsto l(\beta_1, \ldots, \beta_p)$ is quasi-convex and hemivariate. 1) and 2) are obvious and 3) could be shown from Lemma \ref{lem:1}.
\end{proof}

As shown in Theorem 2, the applicability of estimation with LAAD penalty
heavily depends on the range of tuning parameter \(r\), which assures
convergence of the algorithm if \(r \leq 1\). For example, when LAAD
penalization is applied to the stable estimation of loss development
factor, it is known that the estimated loss development factor for later
development years are usually quite low so it is innocuous to impose
such condition, {which is confirmed in Section 4}.

{Finally, it is possible to try extending LAAD regression given in \mbox{(\ref{eq:2})} in various directions. For example, one can consider}

\[
||Y-X\beta||^2 + r \sum_{j=1}^p \log(1 + |\beta_j|^q),
\]
{where $q \in (0, \infty)$ controls the behavior of penalty. If $q = 2$, then we expect the behavior of penalty function is very similar to that of ridge regression for small coefficients. Indeed, we can draw a penalty solution as in Figures 1 and 2 for this case and confirm that while it preserves the shrinkage property, it is unable to perform variable selection. If $q \in (0,1)$, then it is natural to expect stronger variable selection than LAAD, but optimization will be more challenging than that of LAAD, as in usual Bridge regression with $q \in (0,1)$. Further, LAAD penalized regression can be extended to other distributions so that one optimizes the following penalized likelihoods from the GLM family of the form:}
\[
\ell\left(g(X\beta); y \right) + r \sum_{j=1}^p
\log(1 + |\beta_j|),
\] {where $g$ is the link function.}

{However, it should be noted that such forms of objective functions are highly non-convex so that the theoretical results from Theorem 2 do not necessarily hold and corresponding convergence analysis should be performed on a case-by-case. To illustrate, consider a random variable $Y \in \{0, 1\}$ that is distributed as $\mathbb{P}(Y=1|x)=\dfrac{e^{\beta_0+x\beta_1}}{1+e^{\beta_0+x\beta_1}}$; we observe $\mathbf{y}=(0,1)$ and $\mathbf{x}=(1,0)$. In this case, the log-likelihood with LAAD penalty when $r=1$ is given as $\ell (\beta_0, \beta_1)= \beta_0-  \ln(1+e^{\beta_0 + \beta_1}) + \ln (1+|\beta_1|)$. In this case, one can easily check that  $\ell_1: \beta_1 \rightarrow \ell (\beta_0, \beta_1)$ may not be a quasi-convex function even if $r \leq 1$ and $||X||=1$ since $\ell_1(x+\lambda d) > -\ln 2> \max \{\ell_1(x),\, \ell_1(x+d) \}$, when $x=0, d=1, \lambda=0.5$, and $\beta_0=0$.}

\section{Simulation study} \label{sec:simul}

In this section, we conduct a simulation study so as to show the novelty
of the proposed method. Suppose we have the following nine available
covariates \((X_1, \ldots, X_9)\) and response variable \(y\) which are
generated as follows: \[
\begin{aligned}
& X_1 \sim \mathcal{N}(5,1), \ X_2 \sim \mathcal{N}(-2,1), \ X_3 \sim \mathcal{N}(1,4), \ X_4 \sim \mathcal{N}(3,4), \ X_5 \sim \mathcal{N}(0,4), \\
& X_6 \sim \mathcal{N}(0,9), \ X_7 \sim \mathcal{N}(-3,4), \ X_8 \sim \mathcal{N}(2,1), \ X_9 \sim \mathcal{N}(3,1), \ \epsilon \sim \mathcal{N}(0,1), \\
&  y = -X_1+X_2+X_3-X_4+X_5-X_6+X_7+X_8-X_9-10X_1X_6+X_2X_3 + 0.1X_3X_4 - 0.01X_4X_6 + \epsilon,
\end{aligned}
\] so that the simulation scheme can incorporate possible interactions
in the model (while the model is still sparse enough)
{for which the magnitude} effects vary. One can check that if a
regression model is calibrated using \((X_1, X_2, \ldots, X_9)\), then
the estimated regression coefficients are all significant. However, even
if all covariates are significant by themselves, omission of effective
interaction terms can lead to bias in the estimated coefficients and
subsequently lack of fit as illustrated in Figure \ref{fig:4}. In Figure
\ref{fig:4}, reduced model means a linear model fitted only with
\((X_1, \, X_2, \, \ldots, \, X_9)\), while true model is a linear model
fitted with
\((X_1, \, X_2, \, \ldots, \, X_9, \, X_1X_6, \, X_2X_3, \, X_3X_4, \, X_4X_6)\).

\begin{figure}
\centering
\includegraphics{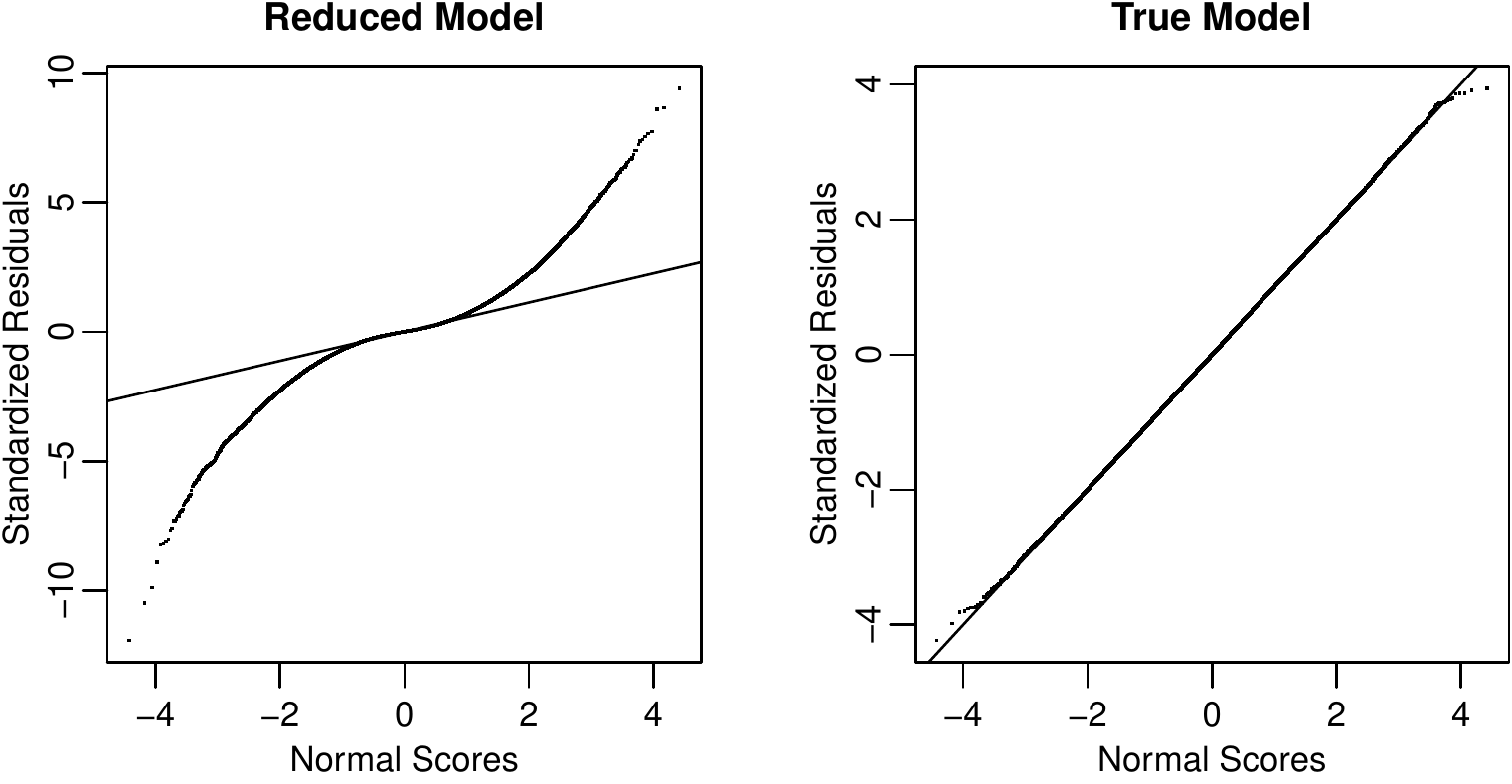}
\caption{\label{fig:4}QQplots for reduced model and true model}
\end{figure}

On the other hand, including every interaction terms also may end up
with {an inferior} model since it may accumulate noise in the
estimation, which leads to {larger} variances in the estimates. As
elaborated in \citet{james2013isl}, the mean squared error (MSE) of a
predicted value under a linear model is determined by both the variance
and the squared bias of the estimated regression coefficients as
follows: \begin{equation}\label{eq:4}
\begin{aligned}
{\mathbb E}\left[y_0-\hat{f}(x_0)\right]^2 &= Var( \hat{f}(x_0)) + [\text{Bias}( \hat{f}(x_0))]^2 +Var(\epsilon) \\
                       &= (x_0)' \left[Var( \hat{\beta}) + \text{Bias}( \hat{\beta})^2  \right](x_0) + \sigma^2.
\end{aligned}
\end{equation}

{As shown in Equation (\mbox{\ref{eq:4}}), prediction performance of a new observation from out-of-sample validation set, $y_0$, can be determined by the stochastic error, $\sigma^2$, and the parameter error, $MSE(\hat{\beta})=Var( \hat{\beta}) + \text{Bias}( \hat{\beta})^2$, in a linear predictive model. Henceforth, it suffices to evaluate the estimation performance of $\hat{\beta}$ to assess the predictive ability of a model since the stochastic error is irreducible and independent of the predictive model.}

{Here} we {also} note that by including fewer variables in our
model with {the} variable selection, we may get lower
\(Var( \hat{\beta})\). However, it could increase
\([\text{Bias}( \hat{\beta})]^2\) due to omitted variable bias (i.e., if
a variable has been selected out) or inherent bias of the estimated
value because of the penalization. Therefore, it implies that if most of
the original variables are significant so that the magnitude of the bias
is too high, then the benefit of a reduced \(Var( \hat{\beta})\) is
compensated by a higher \([\text{Bias}( \hat{\beta})]^2\). In this
regard, variable selection should be performed carefully to make a
balance between the bias and variance and get better prediction with
lower mean squared error.

To show the novelty of our proposed penalty function, we first obtain
100 replications of simulated samples
\((X_1, \, X_2, \, \ldots, \, X_9, \, y)\) with
{sample sizes $100, 300$, and $1000$ to account for the impact of ratio $p/n$ on estimation}
and estimate the regression coefficients based on the following seven
models:

\begin{itemize}
\item[(i)] \textbf{Full model}: a linear model fitted with $(X_1, \, X_2, \, \ldots,  \,X_9)$ and every possible interaction among them,
\item[(ii)]  \textbf{Reduced model}: a linear model fitted only with $(X_1,\,  X_2, \, \ldots,\, X_9)$,
\item[(iii)] \textbf{Best model}: Full model regularized with $L_0$ penalty (forward feature selection),
\item[(vi)]  \textbf{LASSO model}: Full model regularized with $L_1$ penalty,
\item[(v)]   \textbf{MCP model}: Full model regularized with MC penalty,
\item[(vi)]  \textbf{SCAD model}: Full model regularized with SCAD penalty,
\item[(vii)] \textbf{LAAD model}: Full model regularized with LAAD penalty.
\end{itemize}
{All these models were calibrated using the computational routines from the statistical software \texttt{R}. For \textbf{Full model} and \textbf{Reduced model}, the basic \texttt{lm} function was used. \textbf{Best model} was fitted using \texttt{regsubsets} function in \texttt{leaps} package \mbox{\citep{lumley2013package}}. \textbf{LASSO model} was fitted using \texttt{glmnet} in \mbox{\cite{friedman2009glmnet}}. \textbf{MCP model} and \textbf{SCAD model} were fitted using \texttt{plus} package \mbox{\citep{zhang2009plus}}. There is no standard package to solve the \textbf{LAAD model}.}

To evaluate the estimation results under each model, we introduce the
following metrics, which measure the discrepancy between the true
coefficients and estimated coefficients under each model. \[
\begin{aligned}
 \text{Bias for } \beta_{j} &= \frac{1}{100}\sum_{s=1}^{100} (\beta_j - \hat{\beta}_{j(s)}), \ \  \text{Root Mean Squared Error for } \beta_{j} = \sqrt{\frac{1}{100}\sum_{s=1}^{100} (\beta_j - \hat{\beta}_{j(s)})^2}, \\
\end{aligned}
\] where \(\beta_j\) means the true value of \(j^{th}\) coefficient and
\(\hat{\beta}_{j(s)}\) refers to the estimated value of \(j^{th}\)
coefficient with \(s^{th}\) simulated sample. According to Table
\ref{tab:1}, \textbf{Full model} is most favored in terms of the bias of
estimated coefficients, which is reasonable since ordinary least square
(OLS) estimator is unbiased. However, one can see MSEs of estimated
coefficients under \textbf{Full model} tend to be greater than those of
\textbf{LAAD model} so that \textbf{LAAD model} is expected to provide
better estimation,
{especially when $n$, the sample size, is relatively smaller compared to $p$, the number of covariates.}
It is also observed that estimation results with \textbf{Reduced model}
is quite poor whereas the performance of \textbf{Best model} and
\textbf{LAAD model} are the best.

\begin{table}[!h]

\caption{\label{tab:unnamed-chunk-1}\label{tab:1}Summary of estimation}
\centering
\resizebox{\linewidth}{!}{
\begin{tabular}[t]{lrrrrrrrrrrrrrr}
\toprule
\multicolumn{1}{c}{ } & \multicolumn{7}{c}{Bias} & \multicolumn{7}{c}{RMSE} \\
\cmidrule(l{3pt}r{3pt}){2-8} \cmidrule(l{3pt}r{3pt}){9-15}
  & Full & Reduced & Best & LASSO & MCP & SCAD & LAAD & Full & Reduced & Best & LASSO & MCP & SCAD & LAAD\\
\midrule
\addlinespace[0.3em]
\multicolumn{15}{l}{\textbf{Sample size: n=1000}}\\
\hspace{1em}x1 & -0.006 & 0.107 & -0.001 & -0.275 & -0.714 & -0.713 & 0.000 & 0.082 & 0.977 & 0.042 & 0.526 & 0.747 & 0.743 & 0.030\\
\hspace{1em}x2 & -0.020 & 1.080 & -0.002 & 0.490 & -0.938 & -0.916 & -0.002 & 0.165 & 1.399 & 0.056 & 0.930 & 0.970 & 0.957 & 0.040\\
\hspace{1em}x3 & -0.011 & -1.639 & -0.014 & -0.492 & -0.970 & -0.960 & -0.027 & 0.099 & 1.713 & 0.086 & 0.552 & 0.985 & 0.978 & 0.054\\
\hspace{1em}x4 & 0.003 & 0.164 & -0.003 & -0.044 & 0.921 & 0.910 & 0.005 & 0.116 & 0.438 & 0.055 & 0.406 & 0.959 & 0.953 & 0.024\\
\hspace{1em}x5 & 0.012 & -0.048 & 0.006 & -0.178 & -0.133 & -0.036 & -0.014 & 0.123 & 0.512 & 0.042 & 0.241 & 0.363 & 0.131 & 0.030\\
\hspace{1em}x6 & -0.004 & -50.062 & -0.024 & -0.265 & 0.968 & 0.989 & 0.035 & 0.088 & 50.066 & 0.071 & 0.284 & 0.985 & 0.995 & 0.079\\
\hspace{1em}x7 & 0.004 & -0.055 & -0.001 & -0.231 & -0.941 & -0.896 & -0.006 & 0.114 & 0.471 & 0.050 & 0.478 & 0.966 & 0.941 & 0.022\\
\hspace{1em}x8 & -0.003 & 0.045 & -0.001 & -0.250 & -0.953 & -0.966 & -0.084 & 0.180 & 0.822 & 0.135 & 0.746 & 0.975 & 0.980 & 0.246\\
\hspace{1em}x9 & 0.012 & 0.232 & 0.018 & 0.636 & 0.940 & 0.979 & 0.053 & 0.148 & 1.054 & 0.091 & 0.809 & 0.969 & 0.986 & 0.166\\
\hspace{1em}`x1 : x6` & 0.001 & 10.000 & 0.001 & 0.045 & -0.189 & -0.194 & -0.009 & 0.012 & 10.000 & 0.012 & 0.048 & 0.193 & 0.195 & 0.016\\
\hspace{1em}`x2 : x3` & -0.001 & -1.000 & -0.001 & -0.074 & -0.206 & -0.255 & -0.014 & 0.016 & 1.000 & 0.016 & 0.079 & 0.248 & 0.312 & 0.021\\
\hspace{1em}`x3 : x4` & 0.001 & -0.100 & 0.001 & 0.004 & -0.055 & -0.074 & -0.004 & 0.008 & 0.100 & 0.008 & 0.012 & 0.084 & 0.088 & 0.008\\
\hspace{1em}`x4 : x6` & 0.000 & 0.010 & 0.006 & 0.007 & 0.008 & 0.010 & 0.005 & 0.005 & 0.010 & 0.009 & 0.008 & 0.011 & 0.010 & 0.006\\
\hline
\addlinespace[0.3em]
\multicolumn{15}{l}{\textbf{Sample size: n=300}}\\
\hspace{1em}x1 & -0.006 & 0.107 & -0.001 & -0.275 & -0.714 & -0.713 & 0.000 & 0.176 & 1.606 & 0.109 & 0.787 & 0.767 & 0.757 & 0.127\\
\hspace{1em}x2 & -0.020 & 1.080 & -0.002 & 0.490 & -0.938 & -0.916 & -0.002 & 0.370 & 1.900 & 0.328 & 1.308 & 0.974 & 0.990 & 0.377\\
\hspace{1em}x3 & -0.011 & -1.639 & -0.014 & -0.492 & -0.970 & -0.960 & -0.027 & 0.221 & 1.905 & 0.166 & 0.599 & 0.995 & 0.990 & 0.133\\
\hspace{1em}x4 & 0.003 & 0.164 & -0.003 & -0.044 & 0.921 & 0.910 & 0.005 & 0.216 & 0.936 & 0.109 & 0.648 & 0.944 & 0.958 & 0.046\\
\hspace{1em}x5 & 0.012 & -0.048 & 0.006 & -0.178 & -0.133 & -0.036 & -0.014 & 0.259 & 0.903 & 0.116 & 0.395 & 0.594 & 0.409 & 0.120\\
\hspace{1em}x6 & -0.004 & -50.062 & -0.024 & -0.265 & 0.968 & 0.989 & 0.035 & 0.171 & 49.967 & 0.111 & 0.390 & 0.990 & 0.995 & 0.177\\
\hspace{1em}x7 & 0.004 & -0.055 & -0.001 & -0.231 & -0.941 & -0.896 & -0.006 & 0.214 & 0.872 & 0.160 & 0.681 & 0.934 & 0.945 & 0.060\\
\hspace{1em}x8 & -0.003 & 0.045 & -0.001 & -0.250 & -0.953 & -0.966 & -0.084 & 0.397 & 1.873 & 0.490 & 1.332 & 0.975 & 0.985 & 0.540\\
\hspace{1em}x9 & 0.012 & 0.232 & 0.018 & 0.636 & 0.940 & 0.979 & 0.053 & 0.326 & 1.822 & 0.339 & 1.016 & 0.989 & 0.989 & 0.260\\
\hspace{1em}`x1 : x6` & 0.001 & 10.000 & 0.001 & 0.045 & -0.189 & -0.194 & -0.009 & 0.023 & 10.000 & 0.020 & 0.065 & 0.191 & 0.193 & 0.034\\
\hspace{1em}`x2 : x3` & -0.001 & -1.000 & -0.001 & -0.074 & -0.206 & -0.255 & -0.014 & 0.033 & 1.000 & 0.032 & 0.093 & 0.265 & 0.312 & 0.036\\
\hspace{1em}`x3 : x4` & 0.001 & -0.100 & 0.001 & 0.004 & -0.055 & -0.074 & -0.004 & 0.016 & 0.100 & 0.016 & 0.022 & 0.081 & 0.089 & 0.018\\
\hspace{1em}`x4 : x6` & 0.000 & 0.010 & 0.006 & 0.007 & 0.008 & 0.010 & 0.005 & 0.012 & 0.010 & 0.011 & 0.011 & 0.014 & 0.013 & 0.009\\
\hline
\addlinespace[0.3em]
\multicolumn{15}{l}{\textbf{Sample size: n=100}}\\
\hspace{1em}x1 & -0.039 & 0.128 & -0.158 & -0.168 & -0.703 & -0.619 & 0.015 & 0.437 & 3.051 & 0.350 & 1.401 & 0.838 & 0.827 & 0.298\\
\hspace{1em}x2 & -0.040 & 1.191 & -0.305 & 0.170 & -0.910 & -0.830 & -0.305 & 0.923 & 3.366 & 0.677 & 2.148 & 0.982 & 0.997 & 0.584\\
\hspace{1em}x3 & 0.001 & -1.705 & -0.119 & -0.572 & -0.976 & -0.957 & -0.377 & 0.613 & 2.302 & 0.490 & 0.762 & 0.991 & 0.980 & 0.585\\
\hspace{1em}x4 & 0.042 & 0.326 & 0.122 & 0.355 & 0.878 & 0.867 & 0.097 & 0.523 & 1.601 & 0.385 & 0.819 & 0.942 & 0.937 & 0.343\\
\hspace{1em}x5 & -0.149 & 0.078 & -0.203 & -0.520 & -0.730 & -0.560 & -0.330 & 0.612 & 1.568 & 0.507 & 0.702 & 0.856 & 0.745 & 0.516\\
\hspace{1em}x6 & 0.046 & -49.945 & 0.042 & -0.588 & 1.000 & 1.000 & 0.135 & 0.353 & 49.981 & 0.310 & 0.826 & 1.000 & 1.000 & 0.570\\
\hspace{1em}x7 & 0.038 & 0.000 & -0.112 & -0.177 & -0.891 & -0.868 & -0.129 & 0.558 & 1.651 & 0.383 & 1.088 & 0.946 & 0.931 & 0.370\\
\hspace{1em}x8 & 0.043 & 0.155 & -0.621 & -0.489 & -0.978 & -0.983 & -0.606 & 0.910 & 3.275 & 0.888 & 2.066 & 0.987 & 0.990 & 0.751\\
\hspace{1em}x9 & 0.024 & 0.145 & 0.332 & 0.392 & 0.856 & 0.806 & 0.232 & 0.771 & 3.355 & 0.677 & 1.576 & 0.974 & 0.981 & 0.528\\
\hspace{1em}`x1 : x6` & -0.006 & 10.000 & -0.011 & 0.111 & -0.180 & -0.181 & -0.019 & 0.053 & 10.000 & 0.047 & 0.147 & 0.184 & 0.185 & 0.111\\
\hspace{1em}`x2 : x3` & 0.004 & -1.000 & -0.008 & -0.119 & -0.222 & -0.213 & -0.043 & 0.072 & 1.000 & 0.063 & 0.157 & 0.291 & 0.300 & 0.080\\
\hspace{1em}`x3 : x4` & 0.000 & -0.100 & -0.005 & -0.008 & -0.020 & -0.048 & -0.007 & 0.042 & 0.100 & 0.045 & 0.043 & 0.090 & 0.088 & 0.036\\
\hspace{1em}`x4 : x6` & 0.002 & 0.010 & 0.006 & 0.007 & 0.000 & 0.007 & 0.001 & 0.026 & 0.010 & 0.016 & 0.020 & 0.027 & 0.019 & 0.016\\
\hhline{===============}
\end{tabular}}
\end{table}

Besides the values of estimated coefficients, it is also of interest to
capture correct degree of sparsity in a model with the following
measures: \[
\begin{aligned}
\text{Mean } L_1 \text{ norm difference} &= \frac{1}{100}\sum_{s=1}^{100}\sum_{j=1}^p |\beta_j - \hat{\beta}_{j(s)}|, \\
\text{Mean }  L_0 \text{ norm difference} &= \frac{1}{100}\sum_{s=1}^{100}\sum_{j=1}^p \mathbbm{1}_{\{ \beta_j = 0 \ne \hat{\beta}_{j(s)} \text{ or } \beta_j \ne 0 = \hat{\beta}_{j(s)}\}}.
\end{aligned}
\] Table \ref{tab:2} shows how \textbf{LAAD model} captures the sparsity
of the true model correctly.
{Again, it is shown that the estimation efficiency deteriorates as we have smaller sample in all models but one can see that \textbf{LAAD model} and \textbf{Best model} show the smallest mean $L_1$ and $L_0$ norm differences, respectively,}
while \textbf{Full model} fails to capture the sparsity of the true
model. Therefore, this simulation supports the assertion that LAAD
penalty can be utilized in practice with better performance
{in a reasonable amount of computation time relative to other penalization methods.}

\begin{table}[!h]

\caption{\label{tab:unnamed-chunk-2}\label{tab:2}Norm differences and computation times for each model}
\centering
\begin{tabular}[t]{lrrrrrrr}
\toprule
  & Full & Reduced & Best & LASSO & MCP & SCAD & LAAD\\
\midrule
\addlinespace[0.3em]
\multicolumn{8}{l}{\textbf{Mean L1 norm differences}}\\
\hspace{1em}$n=1000$ & 2.325 & 68.429 & 1.578 & 5.961 & 10.048 & 9.888 & 1.450\\
\hspace{1em}$n=300$ & 3.750 & 71.584 & 2.624 & 8.205 & 10.381 & 10.216 & 2.481\\
\hspace{1em}$n=100$ & 7.657 & 78.160 & 5.544 & 12.607 & 11.066 & 10.693 & 5.224\\
\hline
\addlinespace[0.3em]
\multicolumn{8}{l}{\textbf{Mean L0 norm differences}}\\
\hspace{1em}$n=1000$ & 31.000 & 5.000 & 6.550 & 19.930 & 12.900 & 12.410 & 4.510\\
\hspace{1em}$n=300$ & 31.000 & 5.000 & 7.620 & 22.290 & 13.910 & 13.240 & 9.100\\
\hspace{1em}$n=100$ & 31.000 & 5.000 & 10.820 & 23.610 & 15.410 & 15.330 & 13.790\\
\hline
\addlinespace[0.3em]
\multicolumn{8}{l}{\textbf{Average computation times}}\\
\hspace{1em}$n=1000$ & 0.007 & 0.002 & 0.009 & 0.087 & 0.415 & 0.439 & 0.090\\
\hspace{1em}$n=300$ & 0.005 & 0.002 & 0.008 & 0.071 & 0.331 & 0.349 & 0.040\\
\hspace{1em}$n=100$ & 0.005 & 0.002 & 0.007 & 0.067 & 0.296 & 0.229 & 0.024\\
\hhline{========}
\end{tabular}
\end{table}

\section{Empirical application: loss development methods} \label{sec:insurance}

{Claims reserving is a key task to assure solvency of insurer. This section demonstrates an empirical application of using LAAD regression and in spite of the non-convex nature of the penalty, it has the promise of producing stable and smooth estimates of loss development factors, as well as reasonable estimates of insurance claim reserves. For additional application to insurance ratemaking, please see appendix.}

\subsection{Data characteristics} \label{sub:data}

A dataset from ACE Limited 2011 Global Loss Triangles is used for our
empirical analysis which is shown in Tables 3 and 4. This dataset is a
summarization of two lines of insurance business that include
\textbf{General Liability} and \textbf{Other Casualty} in the form of
reported claim triangles.

{The} given dataset can also be expressed {as}:
\begin{equation} \label{eq:trn}
\mathcal{D}_{1:I} = \{Y_{ij}^{(n)}:1 \leq i \leq I \,\text{  and  } \,1 \leq j \leq \min(I, I+1-i),\, n=1,2 \},
\end{equation} where \(Y^{(n)}_{ij}\) {refers to} the reported claim
for \(n^{th}\) line of insurance business in \(i^{th}\) accident years
with \(j^{th}\) development lag. Note that \(I=10\) in our case and
these are displayed in upper-left parts of Tables 3 and 4.

Based on the reported claim data (upper triangle), an insurance company
needs to predict the ultimate claims (lower triangle) described as
follows: \begin{equation} \label{eq:tst}
\mathcal{D}_{I+k} = \{Y_{ij}^{(n)}:1+k \leq i \leq I \,\text{  and  } \, j = I+1+k-i,\, n=1,2 \}.
\end{equation}

\begin{table}[!ht]
\caption{Reported claim triangle for General Liability}
\centering
\begin{tabular}[t]{lrrrrrrrrrr}
\hhline{===========}
  & DL 1 & DL 2 & DL 3 & DL 4 & DL 5 & DL 6 & DL 7 & DL 8 & DL 9 & DL 10\\
\midrule
AY 1 & 87,133 & 146,413 & 330,129 & 417,377 & 456,124 & 556,588 & 563,699 & 570,371 & 598,839 & 607,665\\
AY 2 & 78,132 & 296,891 & 470,464 & 485,708 & 510,283 & 568,528 & 591,838 & 662,023 & 644,021 & \color{blue} 654,481\\
AY 3 & 175,592 & 233,149 & 325,726 & 449,556 & 532,233 & 617,848 & 660,776 & 678,142 & \color{blue} 696,378 & \\
AY 4 & 143,874 & 342,952 & 448,157 & 599,545 & 786,951 & 913,238 & 971,329 & \color{blue} 1,013,749 &  & \\
AY 5 & 140,233 & 284,151 & 424,930 & 599,393 & 680,687 & 770,348 & \color{blue} 820,138 &  &  & \\
AY 6 & 137,492 & 323,953 & 535,326 & 824,561 & 1,056,066 & \color{blue} 1,118,516 &  &  &  & \\
AY 7 & 143,536 & 350,646 & 558,391 & 708,947 & \color{blue} 825,059 &  &  &  &  & \\
AY 8 & 142,149 & 317,203 & 451,810 & \color{blue} 604,155 &  &  &  &  &  & \\
AY 9 & 128,809 & 298,374 & \color{blue} 518,788 &  &  &  &  &  &  & \\
AY 10 & 136,082 & \color{blue} 339,516 &  &  &  &  &  &  &  & \\
\hhline{===========}
\end{tabular}
\end{table}

\begin{table}[!ht]
\caption{Reported claim triangle for Other Casualty }
\centering
\begin{tabular}[t]{lrrrrrrrrrr}
\hhline{===========}
  & DL 1 & DL 2 & DL 3 & DL 4 & DL 5 & DL 6 & DL 7 & DL 8 & DL 9 & DL 10\\
\midrule
AY 1 & 201,702 & 262,233 & 279,314 & 313,632 & 296,073 & 312,315 & 308,072 & 309,532 & 310,710 & 297,929\\
AY 2 & 202,361 & 240,051 & 265,869 & 302,303 & 347,636 & 364,091 & 358,962 & 361,851 & 355,373 & \color{blue} 357,075\\
AY 3 & 243,469 & 289,974 & 343,664 & 360,833 & 372,574 & 373,362 & 382,361 & 380,258 & \color{blue} 384,914  & \\
AY 4 & 338,857 & 359,745 & 391,942 & 411,723 & 430,550 & 442,790 & 437,408 & \color{blue} 438,507  &  & \\
AY 5 & 253,271 & 336,945 & 372,591 & 393,272 & 408,099 & 415,102 & \color{blue} 421,743 &  &  & \\
AY 6 & 247,272 & 347,841 & 392,010 & 425,802 & 430,843 & \color{blue} 455,038 &  &  &  & \\
AY 7 & 411,645 & 612,109 & 651,992 & 688,353 & \color{blue} 711,802 &  &  &  &  & \\
AY 8 & 254,447 & 368,721 & 405,869 & \color{blue} 417,660 &  &  &  &  &  & \\
AY 9 & 373,039 & 494,306 & \color{blue} 550,082 &  &  &  &  &  &  & \\
AY 10 & 453,496 & \color{blue} 618,879 &  &  &  &  &  &  &  & \\
\hhline{===========}
\end{tabular}
\end{table}

\subsection{Model specifications and estimation} \label{sub:model}

In our search for a loss development model, we use cross-classfiied
model which was also introduced in \citet{shi2011depreserve} and
\citet{taylor2016}. For each \(n^{th}\) line of business, unconstrained
lognormal cross-classified model is formulated as follows:
\begin{equation} \label{eq:unconst}
\begin{aligned}
{\mathbb E}\left[ \log  {Y_{ij}^{(n)}}  \right] =\mu_{ij}^{(n)}=\gamma^{(n)}+\alpha^{(n)}_i + \delta^{(n)}_j,
\end{aligned}
\end{equation} where \(\gamma^{(n)}\) means the overall mean of the
losses from \(n^{th}\) line of business, \(\alpha^{(n)}_i\) is the
effect for \(i^{th}\) accident year and \(\delta^{(n)}_j\) means the
cumulative development at \(j^{th}\) year.
{Note that it is customary to use either lognormal
or gamma distribution in the cross-classified loss development model and it has been shown that use of
lognormal distribution has better goodness-of-fit than gamma distribution in Table 4 of \mbox{\citet{jeong2020vinereserving}}, which used the same dataset as in this article.}

\begin{table}[!ht]

\caption{\label{tab:naive}Summary of unconstrained model estimation}
\centering
\begin{tabular}[t]{lrrrr}
\toprule
\multicolumn{1}{c}{ } & \multicolumn{2}{c}{General Liability} & \multicolumn{2}{c}{Other Casualty} \\
\cmidrule(l{3pt}r{3pt}){2-3} \cmidrule(l{3pt}r{3pt}){4-5}
  & Estimate & Pr(>|t|) & Estimate & Pr(>|t|)\\
\midrule
$\gamma$ & 11.382 & 0.000 & 12.173 & 0.000\\
$\delta_2$ & 0.789 & 0.000 & 0.260 & 0.000\\
$\delta_3$ & 1.236 & 0.000 & 0.359 & 0.000\\
$\delta_4$ & 1.515 & 0.000 & 0.430 & 0.000\\
$\delta_5$ & 1.673 & 0.000 & 0.464 & 0.000\\
$\delta_6$ & 1.779 & 0.000 & 0.491 & 0.000\\
$\delta_7$ & 1.825 & 0.000 & 0.489 & 0.000\\
$\delta_8$ & 1.850 & 0.000 & 0.506 & 0.000\\
$\delta_9$ & 1.874 & 0.000 & 0.508 & 0.000\\
$\delta_{10}$ & 1.936 & 0.000 & 0.432 & 0.000\\
\hline
$\alpha_2$ & 0.168 & 0.020 & 0.065 & 0.027\\
$\alpha_3$ & 0.221 & 0.004 & 0.188 & 0.000\\
$\alpha_4$ & 0.505 & 0.000 & 0.370 & 0.000\\
$\alpha_5$ & 0.396 & 0.000 & 0.282 & 0.000\\
$\alpha_6$ & 0.616 & 0.000 & 0.323 & 0.000\\
$\alpha_7$ & 0.570 & 0.000 & 0.835 & 0.000\\
$\alpha_8$ & 0.461 & 0.000 & 0.347 & 0.000\\
$\alpha_9$ & 0.410 & 0.002 & 0.667 & 0.000\\ 
$\alpha_{10}$ & 0.439 & 0.010 & 0.852 & 0.000\\ \hline
Adj-$R^2$ & \multicolumn{2}{c}{1.000}  &  \multicolumn{2}{c}{1.000}  \\
\hhline{=====}
\end{tabular}
\end{table}

It is natural that incremental reported loss amount gradually decreases
while cumulative reported loss amount still increases until it is
developed to ultimate level, which is equivalent to
\(\delta_j \geq \delta_{j'}\) for \(j \geq j'\). It is observed,
however, that estimated values \(\delta_j\) do not show that pattern for
both lines of business in Table \ref{tab:naive}.

In order to handle aforementioned issue, we propose a penalized
cross-classified model. Since both \(\gamma\) and \(\alpha\) are
nuisance parameters in terms of loss development, we modify the
formulation in (\ref{eq:unconst}) in the following {manner}:
\begin{equation} \label{eq:inc}
C_{i,j+1}^{(n)} := \log  \frac{{Y_{i,j+1}^{(n)}}}{Y_{i,j}^{(n)}}  \text{  and  } \ C_{i,j+1}^{(n)} \sim {\mathcal{N}} \left(\zeta^{(n)}_{j+1},\, {\sigma}^2 \right) \  \text{where } \delta^{(n)}_j=\sum_{l=1}^j \zeta^{(n)}_l, \ \zeta^{(n)}_{j+1}=\eta_{j+1} + \kappa^{(n)}_{j+1}.
\end{equation} In this formulation, mean of \(C_{i,l}^{(n)}\),
\(\zeta_l\) can be interpreted as incremental development factor from
\((l-1)^{th}\) year to \(l^{th}\) year so that if \(\zeta_{L+1}=0\) for
a certain value of \(L\), then it implies there is no more development
of loss after \(L\) years of development and \(\zeta_L\) would determine
the tail factor. Therefore, this formulation allows us to choose tail
factor based on {the} variable selection procedure performed with
penalized regression on given data, not by a subjective judgment.
Furthermore, \(\zeta_l\) consists of two parts; \(\eta_l\) which
accounts for the common payment pattern for all lines of business in the
same company, and \(\kappa^{(n)}_l\) which accounts for the specific
payment pattern for each line of business. This approach allows us to
consider possible dependence between the two lines of business in a
{simplified manner}. Those who are interested in more complicated
dependence modeling among different lines of business might refer to
\citet{shi2012multireserve} and \citet{jeong2020vinereserving}.
{Henceforth}, we propose the following six model specifications:

\begin{itemize}

\item \textbf{Unconstrained model:} a model which minimizes the following for all lines of business simultaneously:
$$
 \sum_{i=1}^I \sum_{j=1}^{I-i} \sum_{n=1}^{N}  \left({C_{i,j+1}^{(n)}} - \zeta^{(n)}_{j+1}\right)^2,
$$
\item \textbf{Best subset model:} a model which minimizes Bayesian information criterion (BIC) based on the estimated parameter values.

\item \textbf{LASSO / SCAD / MCP / LAAD constrained models:} models which minimize the following for all lines of business simultaneously with $p_{\lambda}(\cdot)$ as defined in (\ref{eq:penalty}):
$$
\sum_{i=1}^I \sum_{j=1}^{I-i} \sum_{n=1}^{N}    \left({C_{i,j+1}^{(n)}} - \zeta^{(n)}_{j+1}\right)^2 +\left[ \sum_{j=2}^{J-1} \left( p(\eta_{j+1};\lambda) + \sum_{n=1}^{N-1} p_{\lambda}(\kappa^{(n)}_{j+1};\lambda) \right) \right],
$$
\end{itemize}

Although \(\zeta_l\) has been decomposed into two parts (common payment
patterns and line-specific payment patterns), one can also model
\(\zeta_l\) directly for each line of business seperately. Note that for
all constrained models, \(\eta_{2}\) is not penalized in {the}
estimation to avoid underreserving issues
{as a result of regularization} and \(\kappa^{(N)}_{j+1}=0\) for all
\(j\) {to address} the identifiability issue.

When {the} variable selection {through} penalization is
implemented, it is required to set the tuning parameter, which controls
the magnitude of {the} penalty. In search for the tuning parameter
for LASSO / SCAD / MCP / LAAD constrained models, the usual
cross-validation method is applied to choose {the} optimal penalty so
that the average of root mean squared errors (RMSEs) on \(k\)-fold
cross-validation with each value of tuning parameters are examined as
described in \cite{friedman2009glmnet}.
{To avoid overpenalization by choosing a penalty that yields the smallest average of cross-validation RMSEs, we use the geometric average of two penalty values, where the average of cross-validation RMSEs is the smallest and within one standard error of the minimum, respectively. As stated in Theorem 2, it is critical to assure $r \leq 1$ for the convergence of the coordinate descent algorithm with LAAD penalty, which is clearly satisfied here since the optimal $r$ was chosen as $\log(1.005261)$. Figure \mbox{\ref{fig:edf}} shows the effective degree of freedom (edf), the number of non-zero coefficients with LAAD penalization, and the scaled average of cross-validation RMSEs depending on the chosen value of $r$. According to \mbox{\citet{stein1981estimation}} and \mbox{\citet{efron1986edf}}, edf is given as follows:}
\[
\text{df}(r) = \sum_{i=1}^n \frac{ \text{Cov}(y_i, \hat{\mu}_i(r) )}{\sigma^2}  = \nabla \mathbb{E}(\hat{\mu}_i(r) ) = \sum_{i=1}^n \frac{\partial \hat{\mu}_i(r)}{\partial y_i},
\]
{where $\hat{\mu}_i(r)=\mathbf{x}_i' \hat{\beta}(r)$ and $\hat{\beta}(r)$ is an estimate of $\beta$ using $\mathbf{y}=(y_1, \ldots, y_n)$ as the response variable and $r$ determines the degree of penalization. While it is known that edf is the same as the number of non-zero coefficients in LASSO regression \mbox{\citep{zou2007lassoedf}}, the edf in LAAD regression needs to be evaluated empirically. Let $\hat{\beta}^{(i+)}(r)$ and $\hat{\beta}^{(i-)}(r)$ be estimates of $\beta$ using $\mathbf{y}+\epsilon \cdot e_i$ and $\mathbf{y}-\epsilon \cdot e_i$ as the response variables, where $e_i$ is the $i^{th}$ Cartesian coordinate vector and $\epsilon>0$ is an amount of perturbation, respectively. One can then calculate the edf empirically as follows:}
\[
\text{df}(r) = \sum_{i=1}^n \frac{\partial \hat{\mu}_i(r)}{\partial y_i} =  \sum_{i=1}^n \mathbf{x}_i'\frac{\partial \hat{\beta}(r)}{\partial y_i} \simeq \sum_{i=1}^n \mathbf{x}_i' \left( \frac{\hat{\beta}^{(i+)}(r) - \hat{\beta}^{(i-)}(r)}{2\epsilon}  \right).
\]

{The edf that corresponds to the optimal value of $r=\log(1.005261)\simeq 0.00525$ is $12.10501$, which is marked with a blue diamond in Figure \mbox{\ref{fig:edf}}. It also shows that the edf and the number of non-zero coefficients show similar patterns depending on the degree of penalization so that the number of non-zero coefficients might be used as proxy for the effective degree of freedom in a LAAD regression model.}

\begin{figure}
\centering
\includegraphics{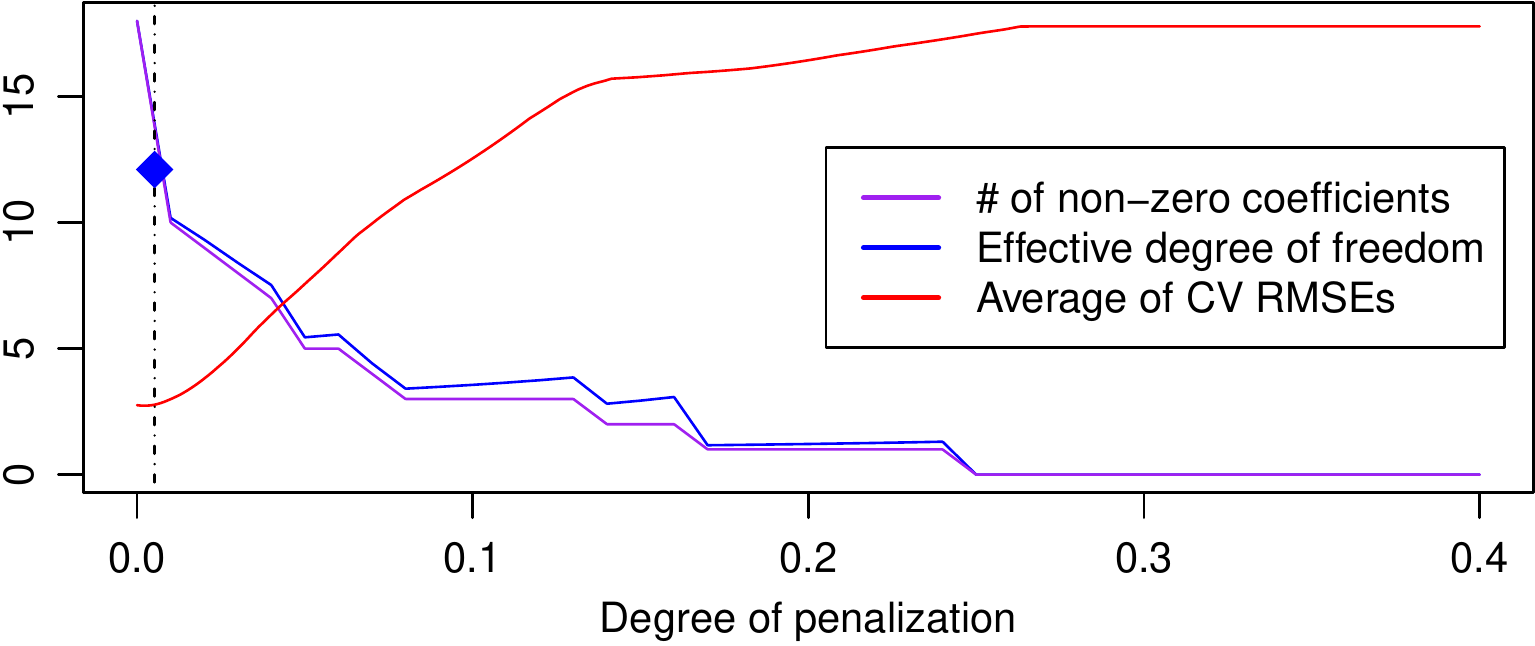}
\caption{\label{fig:edf}Effective degree of freedom with LAAD
penalization}
\end{figure}

Note that not only the choice of tuning parameters, but we also
{might} need to consider different attributes of covariates (for
example, binary, ordinal, discrete, or continuous) when we do {the}
variable selection {with} penalization. However, since the covariates
used in our empirical analysis are all binary variables, we can claim
that either direct use of \(L_1\) penalty or its transformation is
innocuous. For the variable selection on the covariates with diverse
attributes, see \citet{devriendt2018sparse},
{which is reduced to an application of usual LASSO penalty in our case with only binary variables as covariates}.

Once the parameters are estimated in each model, the corresponding
incremental development factor \(j^{th}\) lag for \(n^{th}\) line of
business can be also estimated as \(\exp (\hat{\zeta}^{(n)}_j )\), based
on the formulation of lognormal cross-classified model. Table
\ref{tab:6} summarizes the estimated results of incremental development
factors for the calibrated models. One can see that the unconstrained
model deviates from our expectations on the development pattern. For
example, in the case of General Liability, incremental development
factor of \(7^{th}\) lag is less than that of \(8^{th}\) lag. In the
case of Other Casualty, it is also shown that incremental development
factor of \(9^{th}\) lag is less than 1, which is not intuitive as well.
In contrast, it is observed that all constrained models and best subset
selection model are able to perform {the} variable selection
{and impose smoothness on the sequences of development factors.}

{However, as mentioned in subsection 2.1, use of any penalization method induces bias on the non-zero coefficients while it is not desirable to have huge bias on the non-zero coefficients, due to the imposed shrinkage for smoothing the development factors. More specifically, it is likely that year-to-year development factors are relatively large at the earlier stage of the development. If we use LASSO, the coefficients of early development stage might be underestimated because the soft-thresholding affects all the coefficients regardless of how much a coefficient deviates from zero. This causes bias on the estimated coefficients, which induces less conservative reserve estimates. In Table \mbox{\ref{tab:6}}, it is shown that both LAAD and LASSO regression have the same number of estimated non-zero coefficients but less shrinkage is applied on the coefficients from LAAD regression. Therefore, we think asymptotic unbiasedness of LAAD penalty supports achieving these two tasks simultaneously: smoothing the development factors and avoiding huge bias on the estimated non-zero coefficients.}

\begin{table}[!h]

\caption{\label{tab:unnamed-chunk-4}\label{tab:6}Summary of estimated incremental development factors}
\centering
\resizebox{\linewidth}{!}{
\begin{tabular}[t]{lrrrrrrrrrrrr}
\toprule
\multicolumn{1}{c}{ } & \multicolumn{6}{c}{General Liability} & \multicolumn{6}{c}{Other Casualty} \\
\cmidrule(l{3pt}r{3pt}){2-7} \cmidrule(l{3pt}r{3pt}){8-13}
  & Unconstrained & Best & LASSO & SCAD & MCP & LAAD & Unconstrained & Best & LASSO & SCAD & MCP & LAAD\\
\midrule
$\exp{(\zeta_2)}$ & 2.2022 & 2.3527 & 2.3545 & 2.3067 & 2.2923 & 2.3006 & 1.2975 & 1.3861 & 1.4115 & 1.3590 & 1.3505 & 1.3657\\
$\exp{(\zeta_3)}$ & 1.5681 & 1.5681 & 1.5253 & 1.5681 & 1.5681 & 1.5433 & 1.1052 & 1.1052 & 1.0948 & 1.1052 & 1.1052 & 1.0965\\
$\exp{(\zeta_4)}$ & 1.3108 & 1.3108 & 1.2723 & 1.3108 & 1.3108 & 1.2875 & 1.0792 & 1.0000 & 1.0679 & 1.0508 & 1.0674 & 1.0706\\
$\exp{(\zeta_5)}$ & 1.1723 & 1.1723 & 1.1349 & 1.1723 & 1.1723 & 1.1493 & 1.0352 & 1.0000 & 1.0231 & 1.0000 & 1.0000 & 1.0262\\
$\exp{(\zeta_6)}$ & 1.1569 & 1.1569 & 1.1164 & 1.1569 & 1.1569 & 1.1321 & 1.0298 & 1.0000 & 1.0162 & 1.0000 & 1.0000 & 1.0200\\
$\exp{(\zeta_7)}$ & 1.0465 & 1.0000 & 1.0053 & 1.0022 & 1.0030 & 1.0209 & 0.9959 & 1.0000 & 1.0000 & 1.0000 & 1.0000 & 1.0000\\
$\exp{(\zeta_8)}$ & 1.0512 & 1.0000 & 1.0033 & 1.0000 & 1.0000 & 1.0215 & 1.0024 & 1.0000 & 1.0000 & 1.0000 & 1.0000 & 1.0000\\
$\exp{(\zeta_9)}$ & 1.0106 & 1.0000 & 1.0000 & 1.0000 & 1.0000 & 1.0000 & 0.9929 & 1.0000 & 1.0000 & 1.0000 & 1.0000 & 1.0000\\
$\exp{(\zeta_{10})}$ & 1.0147 & 1.0000 & 1.0000 & 1.0000 & 1.0000 & 1.0000 & 0.9589 & 1.0000 & 1.0000 & 1.0000 & 1.0000 & 1.0000\\
\hhline{=============}
\end{tabular}}
\end{table}

\subsection{Model validation} \label{sub:val}

To validate the predictive models for loss development, calibrated using
the training set (upper loss triangles) \(\mathcal{D}_{1:10}\) defined
in (\ref{eq:trn}), we use cumulative (or incremental) payments of claims
for calendar year 2012 as a validation set, obtained from ACE Limited
2012 Global Loss Triangles. Note that these data points can be described
as
\(\mathcal{D}_{11} = \{Y{_{ij}^{(n)}}: 2 \leq i \leq 10 \text{ and } j = 12-i, n=1,2\}\)
and they are displayed as semi-diagonals in blue color in the triangles
of Tables 3 and 4.

Based on the estimated incremental development factor, one can predict
cumulative (or incremental) payments of claims for the subsequent
calendar year. For example, according to the model specification in
(\ref{eq:inc}), it is possible to predict the cumulative payment for
\(i^{th}\) accident year at \(j+1^{th}\) lag as of \(j^{th}\) lag as
follows: \[
\hat{Y}^{(n)}_{i,j+1} =  Y_{i,j} \times {\mathbb E}\left[Y^{(n)}_{i,j+1}/Y^{(n)}_{i,j}\right] = Y_{i,j} \times {\mathbb E}\left[\exp(C^{(n)}_{i,j})\right] =Y_{i,j} \times \exp\left(\zeta_{j+1}^{(n)}+{\frac{1}{2}\sigma^{2}} \right).
\]

Table \ref{tab:7} provides the predicted values of incremental claims
under each model {as point estimates}. According to the table, we can
see that in case of Other Casualty line, LAAD model is the best for
prediction of total unpaid claims for next calendar year.
{It is also shown that the Unconstrained model fails to perform the variable selection appropriately on the later development factors so that it severely underestimates incremental claims of older accidental years.}
In case of General Liability line, Best / SCAD / MCP models perform
marginally well for the prediction of total unpaid claims, while LASSO
model substantially underestimates the unpaid claims.

\begin{table}[!h]

\caption{\label{tab:unnamed-chunk-5}\label{tab:7}Summary of predicted incremental paid claims}
\centering
\resizebox{\linewidth}{!}{
\begin{tabular}[t]{lrrrrrrrrrrrrrr}
\toprule
\multicolumn{1}{c}{ } & \multicolumn{7}{c}{General Liability} & \multicolumn{7}{c}{Other Casualty} \\
\cmidrule(l{3pt}r{3pt}){2-8} \cmidrule(l{3pt}r{3pt}){9-15}
  & Unconstrained & Best & LASSO & SCAD & MCP & LAAD & Actual & Unconstrained & Best & LASSO & SCAD & MCP & LAAD & Actual\\
\midrule
AY=2004 & 14,647 & 4,986 & 5,301 & 4,803 & 4,746 & 4,932 & 10,460 & -11,930 & 2,751 & 2,925 & 2,650 & 2,619 & 2,722 & 1,702\\
AY=2005 & 12,610 & 5,250 & 5,582 & 5,058 & 4,998 & 5,194 & 18,236 & 275 & 2,944 & 3,130 & 2,836 & 2,803 & 2,912 & 4,655\\
AY=2006 & 57,778 & 7,520 & 11,266 & 7,244 & 7,159 & 28,505 & 42,420 & 4,514 & 3,386 & 3,600 & 3,262 & 3,224 & 3,350 & 1,098\\
AY=2007 & 42,162 & 5,964 & 10,479 & 7,449 & 8,013 & 22,093 & 49,790 & 1,575 & 3,213 & 3,417 & 3,096 & 3,059 & 3,179 & 6,641\\
AY=2008 & 175,372 & 175,194 & 132,625 & 174,848 & 174,740 & 148,629 & 62,450 & 16,338 & 3,335 & 10,599 & 3,213 & 3,175 & 12,000 & 24,195\\
AY=2009 & 128,676 & 128,555 & 102,272 & 128,319 & 128,246 & 112,112 & 116,112 & 29,856 & 5,329 & 21,723 & 5,134 & 5,073 & 23,477 & 23,449\\
AY=2010 & 145,081 & 144,995 & 127,759 & 144,827 & 144,775 & 134,362 & 152,345 & 35,583 & 3,142 & 31,120 & 23,789 & 30,540 & 31,980 & 11,790\\
AY=2011 & 173,204 & 173,136 & 160,477 & 173,003 & 172,962 & 165,626 & 220,413 & 56,300 & 56,220 & 51,295 & 56,065 & 56,017 & 51,845 & 55,776\\
AY=2012 & 165,965 & 186,556 & 186,954 & 180,158 & 178,152 & 179,383 & 203,434 & 139,542 & 179,970 & 191,861 & 167,408 & 163,468 & 170,580 & 165,383\\
\hline
Total & 915,495 & 832,154 & 742,714 & 825,710 & 823,791 & 800,836 & 875,659 & 272,051 & 260,290 & 319,670 & 267,454 & 269,979 & 302,046 & 294,690\\
\hhline{===============}
\end{tabular}}
\end{table}

It is also possible to evaluate the performance of prediction based on
usual validation measures such as root mean squared error (RMSE) and
mean absolute error (MAE) defined as follows: \[
\begin{aligned}
\text{RMSE } =: \sqrt{\frac{1}{9}\sum_{i=2}^{10} (\hat{Y}^{(n)}_{i,12-i} - Y^{(n)}_{i,12-i})^2}, \quad \text{MAE }=: \frac{1}{9}\sum_{i=2}^{10} |\hat{Y}^{(n)}_{i,12-i} - Y^{(n)}_{i,12-i}|.
\end{aligned}
\] Table \ref{tab:8} shows us that LAAD model is the most preferred in
terms of prediction performance measured by RMSE and MAE in both lines
of business. One can see that LAAD model is the best in terms of
out-of-sample validation except for the case of RMSE of General
Liability line, in which LASSO model is the best followed by LAAD model.

\begin{table}[!h]

\caption{\label{tab:unnamed-chunk-6}\label{tab:8}Summary of validation measures}
\centering
\resizebox{\linewidth}{!}{
\begin{tabular}[t]{lrrrrrrrrrrrr}
\toprule
\multicolumn{1}{c}{ } & \multicolumn{6}{c}{General Liability} & \multicolumn{6}{c}{Other Casualty} \\
\cmidrule(l{3pt}r{3pt}){2-7} \cmidrule(l{3pt}r{3pt}){8-13}
  & Unconstrained & Best & LASSO & SCAD & MCP & LAAD & Unconstrained & Best & LASSO & SCAD & MCP & LAAD\\
\midrule
RMSE & 43381.92 & 45677.11 & 36949.02 & 45782.23 & 45827.07 & 37279.38 & 13246.63 & 12032.92 & 10915.02 & 10248.55 & 11335.33 & 8299.45\\
MAE & 27803.04 & 32653.29 & 30366.00 & 33240.06 & 33413.07 & 27464.62 & 10101.76 & 8231.78 & 7903.88 & 6898.50 & 7642.07 & 5557.39\\
\hhline{=============}
\end{tabular}}
\end{table}

{In actuarial practice, it is natural to consider possible ranges of reserve estimates in order to account for random deviation due to model, parameter, and stochastic errors. Despite the prevalent use of standard errors and $t$-test based on the asymptotic properties of maximum likelihood estimates, computation of standard errors in the penalization models has been controversial. There is generally no consensus or agreement on a statistically valid method to calculate standard errors for penalized regression \mbox{\citep{kyung2010penalized}}.  For our purpose, we incorporate possible random deviation of predicted incremental claims using bootstrap methods as in \mbox{\citet{shi2011depreserve}} and \mbox{\citet{gao2018bayesian}}, which can consider both parameter and stochastic errors simultaneously. From Table \mbox{\ref{tab:9}} and Figure \mbox{\ref{fig:5}}, one can see that simulated unpaid claims under each model tend to be centered around the point estimates of total unpaid claims given in Table \mbox{\ref{tab:7}}. All six models have intervals covering the actual unpaid claims. Details of simulation scheme with bootstrap is provided in Appendix C.}

\begin{table}[!h]

\caption{\label{tab:unnamed-chunk-7}\label{tab:9}Interval estimates of predicted incremental claims via bootstrap}
\centering
\begin{tabular}[t]{lrrrrrr}
\toprule
\multicolumn{1}{c}{ } & \multicolumn{3}{c}{General Liability} & \multicolumn{3}{c}{Other Casualty} \\
\cmidrule(l{3pt}r{3pt}){2-4} \cmidrule(l{3pt}r{3pt}){5-7}
  & Mean & 95\% L.I & 95\% U.I & Mean & 95\% L.I & 95\% U.I\\
\midrule
Unconstrained & 915,495 & 658,772 & 1,209,463 & 272,051 & 135,370 & 433,019\\
Best & 742,714 & 574,713 & 1,029,077 & 319,670 & 174,127 & 382,608\\
LASSO & 832,154 & 501,542 & 899,414 & 260,290 & 266,570 & 482,086\\
SCAD & 825,710 & 622,338 & 1,069,123 & 267,454 & 170,353 & 408,123\\
MCP & 823,791 & 605,001 & 1,064,936 & 269,979 & 163,782 & 409,046\\
LAAD & 800,836 & 567,726 & 987,500 & 302,046 & 232,046 & 454,347\\
\hline
Actual & 875,659 &  &  & 294,690 &  & \\
\hhline{=======}
\end{tabular}
\end{table}

\begin{figure}
\centering
\includegraphics{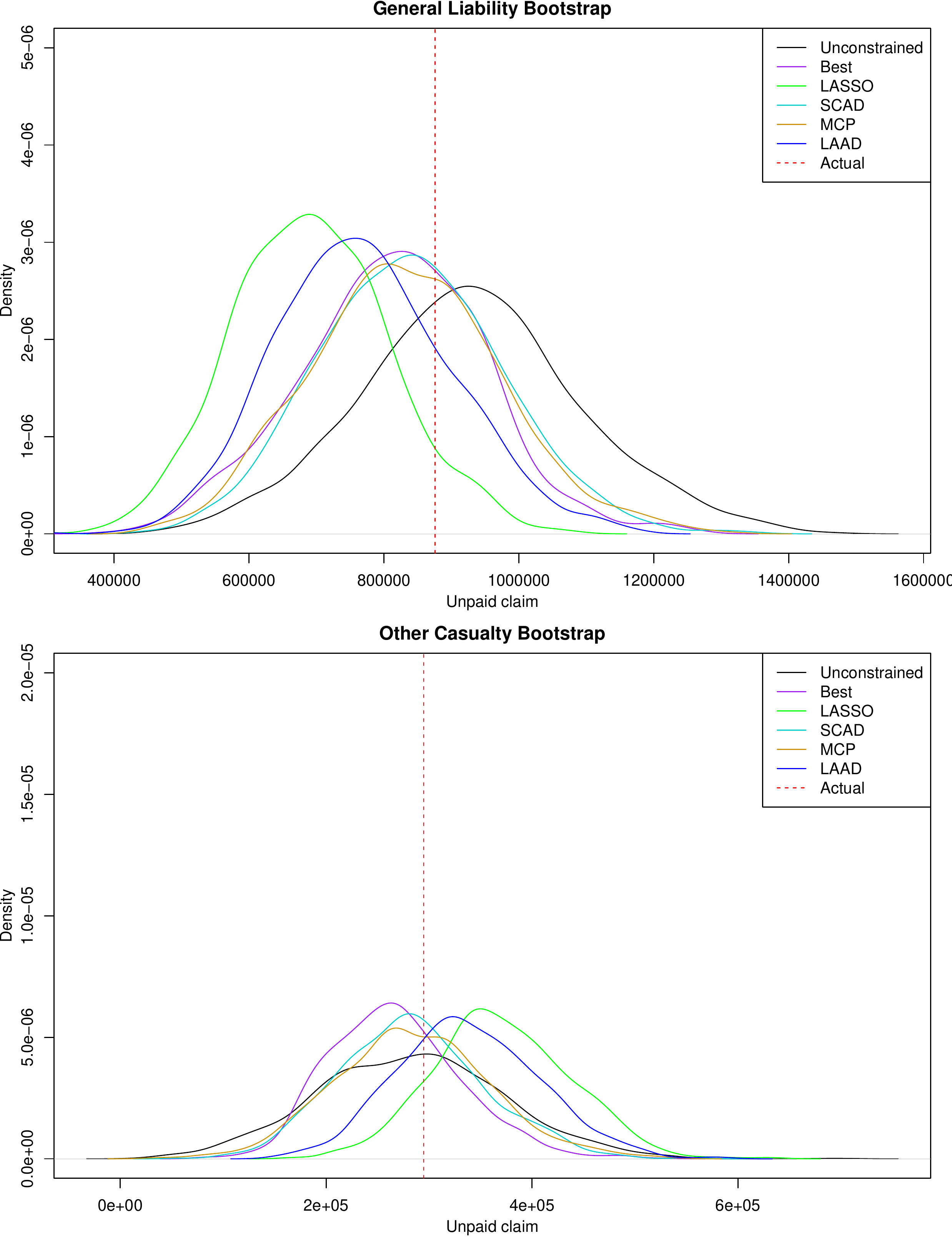}
\caption{\label{fig:5}Predictive density of incremental reported losses
for each model via Bootstrap}
\end{figure}

{Although we focused on the possible application of non-convex penalization in aggregate reserving, we may not preclude use of LAAD penalization to other possible applications in actuarial science and insurance. For the interested readers, please see Appendix D for additional real data analysis, which examines a possible use of non-convex LAAD penalization in insurance ratemaking.}

\section{Concluding remarks} \label{sec:conclude}

In this paper, we introduce use of LAAD penalty {to obtain} stable
estimation of loss development factors. It is also shown that the
proposed penalization method has some desirable properties such as
{the} variable selection with reversion to the true regression
coefficients, analytic solution for the univariate case, and an
optimization algorithm for the multivariate case, which converges under
modest condition {within a coordinate descent algorithm}. The novelty
{and performance} of the proposed method are also shown {using a}
simulation study. In {this} study, {the} use of LAAD regression
outperforms other methods such as OLS {and} other non-convex
penalization in terms of better prediction and the ability to capture
the correct level of model sparsity. {Furthermore,} according to the
results of the empirical {application}, {the} use of LAAD
{regression resulted in a} reasonable loss development pattern with
modest regularization and better prediction of unpaid claims for
{the subsequent} calendar year.
{The results of the use of other non-convex regularization methods, however, are still within tolerance.}
For future research,
{one can extend the use of LAAD regression to more granular loss reserving models},
which would naturally incorporate much more covariates.

\newpage

\hypertarget{appendix-a.-proof-of-theorem-1}{%
\subsection*{Appendix A. Proof of Theorem
1}\label{appendix-a.-proof-of-theorem-1}}
\addcontentsline{toc}{subsection}{Appendix A. Proof of Theorem 1}

It is easy to see that \(\hat{\theta}\times z \geq 0\) so we can start
from the case that \(z\) is not a negative number. Then we have the
following: \[
\begin{aligned}
&\ell'(\theta|r,z)=(\theta-z)+\frac{r}{1+\theta}, \ \ell''(\theta|r,z)=1-\frac{r}{(1+\theta)^2}, \\
&\ell'(\theta^*)=0 \Leftarrow \theta^* = \frac{z-1}{2} + \frac{\sqrt{(z-1)^2+4z-4r}}{2}
\end{aligned}
\] Note that if \(z=r=1\), then
\(\ell''(\theta)=(\theta^2+2\theta)/(1+\theta)^2 > 0\) for \(\theta>0\)
and \(\theta^*=0\). Thus, \(\hat{\theta}=0\).

Case 1) \(z \geq r\)

Since \(\hat{\theta}\) should be non-negative, we just need to consider
\(\theta^* = \frac{z-1}{2} + \frac{\sqrt{(z-1)^2+4z-4r}}{2}\). If
\(z \leq 1\), then \[
\begin{aligned}
&4(1+\theta^*)^2 \geq (z+1+|z-1|)^2 = 2^2 > 4r \ \Rightarrow \ \therefore \ l''(\theta^*|r,z)>0.
\end{aligned}
\] If \(z>1\), then we have \[
\begin{aligned}
&4(1+\theta^*)^2 \geq (z+1+|z-1|)^2 = 4z^2 > 4z \geq 4r \ \Rightarrow \ \therefore \ l''(\theta^*|r,z)>0.
\end{aligned}
\] Thus, for both cases we have only one local minimum point
\(\theta^*\) for \(l'(\theta|r,z)\) and \(\theta^*\) is indeed, a global
minimum point so that
\(\hat{\theta}= \frac{z-1}{2} + \frac{\sqrt{(z-1)^2+4z-4r}}{2}\).

Case 2) \(z < r, z < 1\)

In this case, \(\theta^* < 0\) so that \(l'(\theta|r,z) >0\)
\(\forall \theta \geq 0\). Therefore, \(l(\theta|r,z)\) strictly
increasing and \(\hat{\theta}= 0\).

Case 3) \(r \geq (\frac{z+1}{2})^2\)

In this case, \(\theta^* \notin \mathbb{R}\). Moreover,
\((\frac{z+1}{2})^2 \geq z\), \(l'(0|r,z) =r-z \geq 0\) and
\(l'(\theta|r,z) >0\) \(\forall \theta > 0\). Therefore,
\(\hat{\theta}= 0\).

Case 4) \(1 \leq z < r < (\frac{z+1}{2})^2\)

Here, let \(\theta^{*}=\frac{z-1}{2} + \frac{\sqrt{(z-1)^2+4z-4r}}{2}\)
and \(\theta^{'}=\frac{z-1}{2} - \frac{\sqrt{(z-1)^2+4z-4r}}{2}\). Now,
let us show that \(\theta^{*}\) is the local minimum of
\(\ell(\theta|r,z)\) - which only requires to show that
\(\ell''(\theta^{*}|r,z)>0\). Again, it suffices to show that
\(4(1+\theta^{*})^2 > 4r\) as follows: \[
\begin{aligned}
&4(1+\theta^{*})^2 = (z+1+\sqrt{(z+1)^2-4r})^2 > (z+1)^2+(z+1)^2- 4r > 4r\ \Rightarrow \ \therefore \ \ell''(\theta^{*}|r,z)>0.
\end{aligned}
\] Therefore, \(\theta^{*}\) is a local minimum of \(\ell(\theta|r,z)\)
and \(\hat{\theta}\) would be either \(\theta^{*}\) or \(0\). So in this
case, we have to compute
\(\Delta(z|r) = \ell(\theta^{*}|r,z)-\ell(0|r,z)\) and \[
\hat{\theta} =  \begin{cases}
                   \theta^{*} \ , \quad \text{if } \Delta(z|r) < 0, \\
                  0   \quad,      \quad \text{if } \Delta(z|r) > 0
                \end{cases}
\] Note that for fixed \(r\), \[
\begin{aligned}
\Delta(z|r)  &= \frac{1}{2}(\theta^{*})^2-\theta^{*}z+r\log(1+\theta^{*}), \\
\Delta'(z|r) &= (\theta^{*}-z+\frac{r}{1+\theta^{*}})\frac{\partial \theta^{*}}{\partial z} -\theta^{*} \\
             &= -\theta^{*} \quad (\because \ l(\theta^{*}|r,z)=\theta^{*}-z+\frac{r}{1+\theta^{*}}=0)
\end{aligned}
\] Thus, \(\Delta(z|r)\) is strictly decreasing with respect to \(z\)
and \[
\begin{aligned}
\Delta(z|r) &= \frac{1}{2}(\theta^{*})^2-\theta^{*}z+r\log(1+\theta^{*}) \\
       &= \frac{1}{2}\left(\frac{z-1}{2} + \frac{\sqrt{(z-1)^2+4z-4r}}{2}\right)^2 -\left(\frac{z-1}{2}  \frac{\sqrt{(z-1)^2+4z-4r}}{2}\right)z \\
       & \quad + r \log\left(\frac{z+1}{2} + \frac{\sqrt{(z-1)^2+4z-4r}}{2}\right) = 0
\end{aligned}
\] has unique solution because
\(\Delta(z|r)<0 \Leftrightarrow \hat{\theta}=\theta^{*}\) if \(z=r\) and
\(\Delta(z|r)>0 \Leftrightarrow \hat{\theta}=0\) if \(z=2\sqrt{r}-1\).
Hence \[
\hat{\theta} = \left(\frac{z-1}{2} + \frac{\sqrt{(z-1)^2+4z-4r}}{2}\right)(\mathbbm{1}_{ \{z\geq z^*(r) \}}).
\] where \(z^*(r)\) is the unique solution of \(\Delta(z|r)=0\) for
given \(r\). See Figure \ref{fig:7}.

\begin{figure}
\centering
\includegraphics{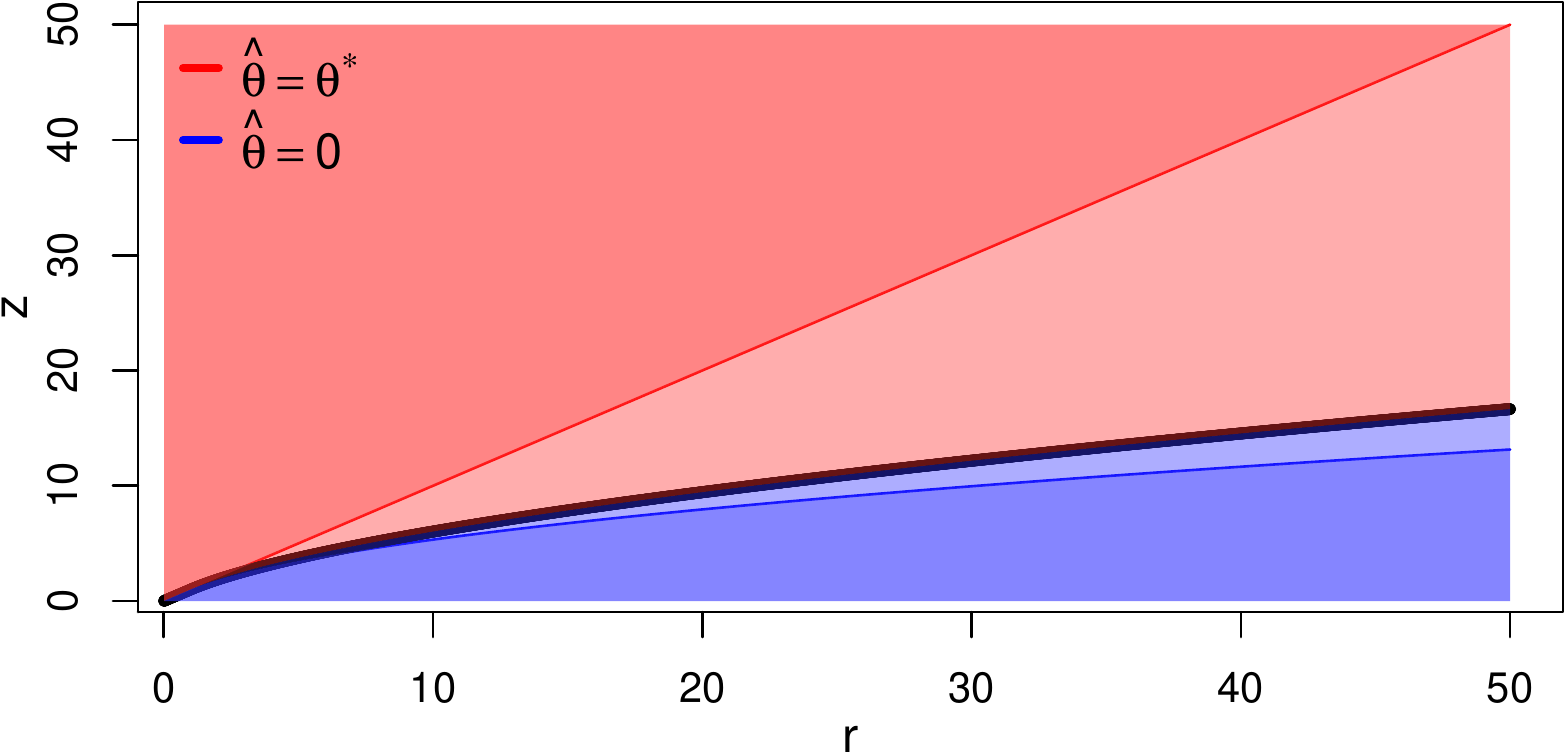}
\caption{\label{fig:7}Distribution of optimizer along with r and z}
\end{figure}

Once we get a result for \(z \geq 0\), we can use the same approach to
\(l(\theta|r,-z)\) when \(z<0\).

\hypertarget{appendix-b.-proof-of-lemma-1}{%
\subsection*{Appendix B. Proof of Lemma
1}\label{appendix-b.-proof-of-lemma-1}}
\addcontentsline{toc}{subsection}{Appendix B. Proof of Lemma 1}

Suppose \(\beta_j\) is fixed as \(w_j\) for all
\(j=1,\ldots, k-1,k+1, \ldots, p\). Then we can observe that \[
\begin{aligned}
\ell_k(\theta) &= \frac{1}{2} \bigg|\bigg| (y-\sum_{j\ne k}X_j w_j)-X_k\theta \bigg|\bigg|^2 + r \log(1+|\theta|) + r\sum_{j \ne k} \log(1+|w_j|)  \\
              &=\frac{1}{2} || t_k -X_k\theta||^2+r \log(1+|\theta|) + r\sum_{j \ne k} \log(1+|w_j|)  \\
              &=\frac{1}{2} (\theta'X'_kX_k\theta -2\theta X_k't_k+t_k't_k) + r \log(1+|\theta|) + r\sum_{j \ne k} \log(1+|w_j|)  \\
              &=\frac{1}{2} (\theta-z_k)^2 + r \log(1+|\theta|) + C_k  \\
\end{aligned}
\] where
\(C_k =\frac{1}{2}(t_k't_k-t_k'X_k X_k't_k) + r\sum_{j \ne k} \log(1+|w_j|)\)
and \(z_k=X_k't_k\).\\

As usual, we can start from the case that \(z_k\geq 0\). First, one can
easily check that \(\ell_k(\theta)\) is a decreasing function of
\(\theta\) where \(\theta \leq0\) and \(z_k\geq 0\). When
\(\theta > 0\), according to the arguments in the proof of Theorem 1,
\(l_k(\theta)\) is strictly decreasing when \(\theta \in (0,\theta^*]\)
and strictly increasing when \(\theta \in [\theta^*,\infty)\) if \(r\)
and \(z_k\) belong to Case 1, Case 2, and Case 3. Note that if
\(r\leq 1\), then we may exclude Case 4. Therefore, \(l_k(\theta)\) is
hemivariate and quasi-convex if \(z_k\geq 0\) and also if \(z_k < 0\)
because of the symmetry of penalty term.

\hypertarget{appendix-c.-bootstrap-for-predictive-distribution-of-unpaid-loss}{%
\subsection*{Appendix C. Bootstrap for predictive distribution of unpaid
loss}\label{appendix-c.-bootstrap-for-predictive-distribution-of-unpaid-loss}}
\addcontentsline{toc}{subsection}{Appendix C. Bootstrap for predictive
distribution of unpaid loss}

\begin{itemize}
  \item[(1)] Simulate $\{\hat{c}^{(n)}_{ij[s]} |\, i=1,\ldots, I,\, j=1,\ldots J-I+1, \,n=1,2\}$ where $\log \hat{c}^{(n)}_{ij[s]} \sim \mathcal{N}(\hat{\eta}_{ij}^{(n)}, \hat{\sigma}^{2})$.
  \item[(2)] Using the simulated values of $\hat{c}^{(n)}_{ij[s]}$ in step (1), estimate bootstrap replication of the parameters $\{(\hat{\eta}_{ij[s]}^{(n)}, \hat{\sigma}_{[s]}^{2})|\,i=1,\ldots, I,\, j=1,\ldots J-I+1, \,n=1,2 \}$.
  \item[(3)] Based on $(\hat{\eta}_{ij[s]}^{(n)}, \hat{\sigma}_{[s]}^{2})$, predict the unpaid loss ${L}^{(n)}$ for the next year which is given as follows:
$$
\hat{L}^{(n)}_{[s]} =\sum_{i=2}^{10} \left(\exp{(\hat{\eta}_{i,12-i[s]}^{(n)}+\hat{\sigma}_{[s]}^{2}/2)}-1\right) y^{(n)}_{i,11-i} 
$$  
Note that the values of $y^{(n)}_{i,11-i}$ for $i=2, \ldots, 10$ and $n=1,2$ are already known in advance from the training set.
  \item[(4)] Repeat steps (1), (2), and (3) for $s=1, \ldots, S$ to obtain the predictive distribution and standard error of ${L}^{(n)}$.
\end{itemize}

\pagebreak

\hypertarget{appendix-d.-application-of-laad-penalty-in-insurance-ratemaking}{%
\subsection*{Appendix D. Application of LAAD penalty in insurance
ratemaking}\label{appendix-d.-application-of-laad-penalty-in-insurance-ratemaking}}
\addcontentsline{toc}{subsection}{Appendix D. Application of LAAD
penalty in insurance ratemaking}

To illustrate possible use of LAAD penalized regression for insurance
ratemaking, we used the LGPIF (Wisconsin Local Government Property
Insurance Fund) data, which consists of policy characteristics and
claims information on multiple lines of business. Each observation is a
local government unit such as city, town, village, or county. For
simplicity, among the multiple types of claims, we extracted the
information from IM (inland marine) line of business only. We used 235
observations with positive claim amounts with corresponding policy
characteristics summarized in Table \ref{tab:LGPIF}.

\begin{table}[!ht]
\begin{center}
\caption{Observable policy characteristics used as covariates} \label{tab:LGPIF}
\resizebox{!}{3cm}{
\begin{tabular}{l|lrrr}
\hline \hline
Categorical & Description &  & \multicolumn{2}{c}{Proportions} \\
variables \\
\hline
TypeCity & Indicator for city entity:           & Y=1 & \multicolumn{2}{c}{24.26 \%} \\
TypeCounty & Indicator for county entity:       & Y=1 & \multicolumn{2}{c}{42.55 \%} \\
TypeMisc & Indicator for miscellaneous entity:  & Y=1 & \multicolumn{2}{c}{0.43 \%} \\
TypeSchool & Indicator for school entity:       & Y=1 & \multicolumn{2}{c}{5.96 \%} \\
TypeTown & Indicator for town entity:           & Y=1 & \multicolumn{2}{c}{8.08 \%} \\
TypeVillage & Indicator for village entity:     & Y=1 & \multicolumn{2}{c}{18.72 \%} \\
NoClaimCreditIM & No IM claim {in three consecutive prior years}:    & Y=1 & \multicolumn{2}{c}{37.02 \%} \\
\hline
 Continuous & & Minimum & Mean & Maximum \\
 variables \\
\hline
CoverageIM  & Log coverage amount of IM claim in mm  &  0.02 & 5.05
            & 46.75\\
lnDeductIM  & Log deductible amount for IM claim     &  6.215 & 6.751
            & 8.517\\
\hline \hline
\end{tabular}}
\end{center}
\end{table}

Since the original dataset is quite rich for data analysis, one can
consider several different aspects of ratemaking, such as the heavy-tail
behavior of claims, possible dependence between frequency and severity,
and serial dependence among the claims of the same policyholder when
observed over time. We do not consider the aformentioned topics to avoid
distraction, but we refer the interested readers to
\citet{frees2016multivariate}, \citet{lee2019dependent}, and
\citet{jeong2020bregman}.

The observed total positive claim amount, \(S_i\), can be described as
follows: \[
\ln S_{i} \sim {\mathcal{N}} \left(\mu_{i},\sigma^2 \right) \ \text{  and  } \mu_{i} = \mathbf{x}_i\beta, 
\] and six different competing models were examined for calibration.

\begin{itemize}

\item \textbf{Unconstrained model:} a model which minimizes the following objective function:
$$
 \sum_{i=1}^m   \left(\ln S_{i} - \mathbf{x}_i\beta\right)^2,
$$
\item \textbf{Best subset model:} a model which minimizes Bayesian information criterion (BIC) based on the estimated parameter values.

\item \textbf{LASSO / SCAD / MCP / LAAD constrained models:} models which minimize the following objective functions with $p_{\lambda}(\cdot)$ as defined in (\ref{eq:penalty}):
$$
 \sum_{i=1}^m   \left(\ln S_{i} - \mathbf{x}_i\beta\right)^2 +\left[ \sum_{k=1}^{p}  p_\lambda\left(\beta_{k} \right)  \right],
$$
\end{itemize}

As usual, note that for all constrained models, we do not impose penalty
constraint on \(\beta_{0}\), the intercept coefficient.

Table \ref{tab:11} summarizes the estimated coefficients of ratemaking
factors. One can see that originally we have six types of location
variables but we end up with three types of location variables (City,
County, and Others) after variable selection is performed using
regularization methods. One can also immediately deduce an intuitive
explanation to these results. For example, City and County are some of
the largest government entities. Coincidentally, the estimated
coefficients happen to be identical for Best, SCAD, and MCP models.

\begin{table}[!h]

\caption{\label{tab:unnamed-chunk-9}\label{tab:11}Summary of estimated ratemaking factors}
\centering
\begin{tabular}[t]{lrrrrrr}
\toprule
  & Unconstrained & Best & LASSO & SCAD & MCP & LAAD\\
\midrule
TypeCity & 1.1342 & 1.0352 & 0.6890 & 1.0352 & 1.0352 & 0.5655\\
TypeCounty & 0.7239 & 0.6125 & 0.3243 & 0.6125 & 0.6125 & 0.1851\\
TypeMisc & -0.5596 & - & - & - & - & -\\
TypeSchool & 0.3957 & - & - & - & - & -\\
TypeTown & 0.0397 & - & - & - & - & -\\
CoverageIM & 0.0497 & 0.0483 & 0.0460 & 0.0483 & 0.0483 & 0.0486\\
lnDeductIM & -0.0246 & - & - & - & - & -\\
NoClaimCreditIM & 0.0839 & - & - & - & - & -\\
\hhline{=======}
\end{tabular}
\end{table}

Table \ref{tab:12} shows that LAAD regression is the best performing
model in terms of out-of-sample validation measures, and followed by
LASSO regression. We conclude that regularization methods can also be
applied to insurance ratemaking problems so that one can effectively
reduce the dimension of ratemaking factors that can be used and still
have better prediction performance.

\begin{table}[!h]

\caption{\label{tab:unnamed-chunk-10}\label{tab:12}Summary of validation measures}
\centering
\begin{tabular}[t]{lrrrrrr}
\toprule
  & Unconstrained & Best & LASSO & SCAD & MCP & LAAD\\
\midrule
RMSE & 22039.11 & 23029.08 & 16895.38 & 23029.08 & 23029.08 & 15498.43\\
MAE & 19330.80 & 20393.52 & 15504.69 & 20393.52 & 20393.52 & 14279.56\\
\hhline{=======}
\end{tabular}
\end{table}

\bibliography{BayesianLASSO}

\begin{thebibliography}{}

\bibitem[Armagan et~al., 2013]{armagan2013gdp}
Armagan, A., Dunson, D.~B., and Lee, J. (2013).
\newblock Generalized double {Pareto} shrinkage.
\newblock {\em Statistica Sinica}, 23(1):119.

\bibitem[Chan et~al., 1982]{chan1982wh}
Chan, F., Chan, L., and Mead, E. (1982).
\newblock Properties and modifications of {W}hittaker-{H}enderson graduation.
\newblock {\em Scandinavian Actuarial Journal}, 1982:57--61.

\bibitem[Devriendt et~al., 2018]{devriendt2018sparse}
Devriendt, S., Antonio, K., Reynkens, T., and Verbelen, R. (2018).
\newblock Sparse regression with multi-type regularized feature modeling.
\newblock {\em arXiv preprint arXiv:1810.03136}.

\bibitem[Efron, 1986]{efron1986edf}
Efron, B. (1986).
\newblock How biased is the apparent error rate of a prediction rule?
\newblock {\em Journal of the American Statistical Association},
  81(394):461--470.

\bibitem[Fan and Li, 2001]{fan2001scad}
Fan, J. and Li, R. (2001).
\newblock Variable selection via nonconcave penalized likelihood and its oracle
  properties.
\newblock {\em Journal of the American Statistical Association},
  96(456):1348--1360.

\bibitem[Frees et~al., 2016]{frees2016multivariate}
Frees, E.~W., Lee, G., and Yang, L. (2016).
\newblock Multivariate frequency-severity regression models in insurance.
\newblock {\em Risks}, 4(1):4.

\bibitem[Friedman et~al., 2009]{friedman2009glmnet}
Friedman, J., Hastie, T., and Tibshirani, R. (2009).
\newblock glmnet: Lasso and elastic-net regularized generalized linear models.
\newblock {\em R package version 4.0-2}.

\bibitem[Gao, 2018]{gao2018bayesian}
Gao, G. (2018).
\newblock {\em Bayesian Claims Reserving Methods in Non-life Insurance with
  Stan}.
\newblock Springer.

\bibitem[Hastie et~al., 2009]{hastie2009elements}
Hastie, T., Tibshirani, R., and Friedman, J. (2009).
\newblock {\em The elements of statistical learning: data mining, inference,
  and prediction}.
\newblock Springer Science \& Business Media.

\bibitem[Hastie et~al., 2015]{hastie2015sparsity}
Hastie, T., Tibshirani, R., and Wainwright, M. (2015).
\newblock {\em Statistical learning with sparsity: the lasso and
  generalizations}.
\newblock CRC press.

\bibitem[Hoerl and Kennard, 1970]{hoerl1970ridge}
Hoerl, A.~E. and Kennard, R.~W. (1970).
\newblock Ridge regression: Biased estimation for nonorthogonal problems.
\newblock {\em Technometrics}, 12(1):55--67.

\bibitem[James et~al., 2013]{james2013isl}
James, G., Witten, D., Hastie, T., and Tibshirani, R. (2013).
\newblock {\em An Introduction to Statistical Learning}.
\newblock Springer.

\bibitem[Jeong, 2020]{jeong2020bregman}
Jeong, H. (2020).
\newblock Testing for random effects in compound risk models via {B}regman
  divergence.
\newblock {\em ASTIN Bulletin: The Journal of the IAA}, 50(3):777--798.

\bibitem[Jeong and Dey, 2020]{jeong2020vinereserving}
Jeong, H. and Dey, D.~K. (2020).
\newblock Application of vine copula for multi-line insurance reserving.
\newblock {\em Risks}, 8(4):111.

\bibitem[Knight and Fu, 2000]{knight2000asymptotics}
Knight, K. and Fu, W. (2000).
\newblock Asymptotics for lasso-type estimators.
\newblock {\em The Annals of Statistics}, 28(5):1356--1378.

\bibitem[Kyung et~al., 2010]{kyung2010penalized}
Kyung, M., Gill, J., Ghosh, M., and Casella, G. (2010).
\newblock Penalized regression, standard errors, and {B}ayesian lassos.
\newblock {\em Bayesian Analysis}, 5(2):369--411.

\bibitem[Lee et~al., 2010]{lee2010hal}
Lee, A., Caron, F., Doucet, A., and Holmes, C. (2010).
\newblock A hierarchical {Bayesian} framework for constructing
  sparsity-inducing priors.
\newblock {\em arXiv preprint arXiv:1009.1914}.

\bibitem[Lee and Shi, 2019]{lee2019dependent}
Lee, G.~Y. and Shi, P. (2019).
\newblock A dependent frequency--severity approach to modeling longitudinal
  insurance claims.
\newblock {\em Insurance: Mathematics and Economics}, 87:115--129.

\bibitem[Lumley, 2013]{lumley2013package}
Lumley, T. (2013).
\newblock R package ‘leaps’.
\newblock {\em R package version 3.1}.

\bibitem[Luo and Tseng, 1992]{luo1992coordinate}
Luo, Z.-Q. and Tseng, P. (1992).
\newblock On the convergence of the coordinate descent method for convex
  differentiable minimization.
\newblock {\em Journal of Optimization Theory and Applications}, 72(1):7--35.

\bibitem[Mack, 1993]{mack1993chainladder}
Mack, T. (1993).
\newblock Distribution-free calculation of the standard error of chain ladder
  reserve estimates.
\newblock {\em ASTIN Bulletin: The Journal of the IAA}, 23(2):213--225.

\bibitem[Mack, 1999]{mack1999chainladderse}
Mack, T. (1999).
\newblock The standard error of chain ladder reserve estimates: Recursive
  calculation and inclusion of a tail factor.
\newblock {\em ASTIN Bulletin: The Journal of the IAA}, 29(2):361--366.

\bibitem[Mazumder et~al., 2011]{mazumder2011sparsenet}
Mazumder, R., Friedman, J.~H., and Hastie, T. (2011).
\newblock Sparsenet: Coordinate descent with nonconvex penalties.
\newblock {\em Journal of the American Statistical Association},
  106(495):1125--1138.

\bibitem[McGuire et~al., 2018]{mcguire2018lasso}
McGuire, G., Taylor, G., and Miller, H. (2018).
\newblock Self-assembling insurance claim models using regularized regression
  and machine learning.
\newblock {\em Available at SSRN 3241906}.

\bibitem[Nawar, 2016]{nawar2016}
Nawar, S.~M. (2016).
\newblock Machine learning techniques for detecting hierarchical interactions
  in insurance claims models.
\newblock Master's thesis, Concordia University.

\bibitem[Park and Casella, 2008]{park2008bayesian}
Park, T. and Casella, G. (2008).
\newblock The {Bayesian LASSO}.
\newblock {\em Journal of the American Statistical Association},
  103(482):681--686.

\bibitem[Renshaw, 1989]{renshaw1989chain}
Renshaw, A.~E. (1989).
\newblock Chain ladder and interactive modelling.(claims reserving and glim).
\newblock {\em Journal of the Institute of Actuaries}, 116(3):559--587.

\bibitem[Shi et~al., 2012]{shi2012multireserve}
Shi, P., Basu, S., and Meyers, G.~G. (2012).
\newblock A bayesian log-normal model for multivariate loss reserving.
\newblock {\em North American Actuarial Journal}, 16(1):29--51.

\bibitem[Shi and Frees, 2011]{shi2011depreserve}
Shi, P. and Frees, E.~W. (2011).
\newblock Dependent loss reserving using copulas.
\newblock {\em ASTIN Bulletin: The Journal of the IAA}, 41(2):449--486.

\bibitem[Stein, 1981]{stein1981estimation}
Stein, C.~M. (1981).
\newblock Estimation of the mean of a multivariate normal distribution.
\newblock {\em The Annals of Statistics}, 9(6):1135--1151.

\bibitem[Taylor and McGuire, 2016]{taylor2016}
Taylor, G. and McGuire, G. (2016).
\newblock {\em Stochastic loss reserving using generalized linear models}.
\newblock Casualty {A}ctuarial {S}ociety.

\bibitem[Tibshirani, 1996]{tibshirani1996lasso}
Tibshirani, R. (1996).
\newblock Regression shrinkage and selection via the {LASSO}.
\newblock {\em Journal of the Royal Statistical Society. Series B
  (Methodological)}, 58(1):267--288.

\bibitem[Tseng, 2001]{tseng2001cdescent}
Tseng, P. (2001).
\newblock Convergence of a block coordinate descent method for
  nondifferentiable minimization.
\newblock {\em Journal of Optimization Theory and Applications},
  109(3):475--494.

\bibitem[Wang et~al., 2019]{wang2019novel}
Wang, S., Zhang, H., Chai, H., and Liang, Y. (2019).
\newblock A novel log penalty in a path seeking scheme for biomarker selection.
\newblock {\em Technology and Health Care}, 27(S1):85--93.

\bibitem[Williams et~al., 2015]{williams2015elastic}
Williams, B., Hansen, G., Baraban, A., and Santoni, A. (2015).
\newblock A practical approach to variable selection—a comparison of various
  techniques.
\newblock In {\em Casualty Actuarial Society E-Forum}, pages 4--40.

\bibitem[Yin and Lin, 2016]{yin2016efficient}
Yin, C. and Lin, X.~S. (2016).
\newblock Efficient estimation of erlang mixtures using iscad penalty with
  insurance application.
\newblock {\em ASTIN Bulletin: The Journal of the IAA}, 46(3):779--799.

\bibitem[Zhang, 2010]{zhang2010mcp}
Zhang, C.-H. (2010).
\newblock Nearly unbiased variable selection under minimax concave penalty.
\newblock {\em The Annals of Statistics}, 38(2):894--942.

\bibitem[Zhang and Melnik, 2009]{zhang2009plus}
Zhang, C.-H. and Melnik, O. (2009).
\newblock plus: Penalized linear unbiased selection.
\newblock {\em R package version 1.0}.

\bibitem[Zou et~al., 2007]{zou2007lassoedf}
Zou, H., Hastie, T., and Tibshirani, R. (2007).
\newblock On the “degrees of freedom” of the lasso.
\newblock {\em The Annals of Statistics}, 35(5):2173--2192.

\end{thebibliography}

\end{document}